\documentclass[12pt,english]{amsart}
\usepackage[T1]{fontenc}
\usepackage[latin9]{inputenc}
\usepackage{geometry}
\usepackage{xcolor}
\usepackage{babel}
\usepackage{amsbsy}
\usepackage{amstext}
\usepackage{amsthm}
\usepackage{amssymb}
\usepackage{graphicx}
\usepackage[unicode=true,pdfusetitle,
 bookmarks=true,bookmarksnumbered=false,bookmarksopen=false,
 breaklinks=true,pdfborder={0 0 1},backref=false,colorlinks=false]
 {hyperref}

\makeatletter

\newcommand*\LyXZeroWidthSpace{\hspace{0pt}}

\numberwithin{equation}{section}
\numberwithin{figure}{section}
\theoremstyle{plain}
\newtheorem{thm}{\protect\theoremname}[section]
\theoremstyle{remark}
\newtheorem{rem}[thm]{\protect\remarkname}
\theoremstyle{definition}
\newtheorem{defn}[thm]{\protect\definitionname}
\theoremstyle{plain}
\newtheorem{lem}[thm]{\protect\lemmaname}
\theoremstyle{plain}
\newtheorem{assumption}[thm]{\protect\assumptionname}
\theoremstyle{plain}
\newtheorem{prop}[thm]{\protect\propositionname}
\theoremstyle{remark}
\newtheorem*{acknowledgement*}{\protect\acknowledgementname}

\usepackage{xurl}

\newcommand*{\rom}[1]{\expandafter\@slowromancap\romannumeral #1@}
\newcommand*{\ara}[1]{\expandafter\@slowromancap\arabic #1@}

\makeatother

\providecommand{\acknowledgementname}{Acknowledgement}
\providecommand{\assumptionname}{Assumption}
\providecommand{\definitionname}{Definition}
\providecommand{\lemmaname}{Lemma}
\providecommand{\propositionname}{Proposition}
\providecommand{\remarkname}{Remark}
\providecommand{\theoremname}{Theorem}

\begin{document}
\global\long\def\B{\mathcal{B}}%

\global\long\def\R{\mathbb{R}}%

\global\long\def\Q{\mathbb{Q}}%

\global\long\def\Z{\mathbb{Z}}%

\global\long\def\C{\mathbb{C}}%

\global\long\def\N{\mathbb{N}}%

\global\long\def\RR{\mathbb{\overline{R}}}%

\global\long\def\rme{\mathrm{e}}%

\global\long\def\rmi{\mathrm{i}}%

\global\long\def\rmd{\mathrm{d}}%

\global\long\def\E{\mathcal{E}}%

\global\long\def\V{\mathcal{V}}%

\global\long\def\P{\mathcal{P}}%

\global\long\def\VR{\mathcal{V}_{\mathcal{R}}}%

\global\long\def\dom{\Omega}%

\global\long\def\sdom{\mathcal{S}}%

\global\long\def\spec#1{\textrm{\text{Spec}}\left(#1\right)}%

\global\long\def\meig#1{\mathcal{\lambda}^{\Gamma}\left(#1\right)}%

\global\long\def\deig#1{\lambda\left(#1\right)}%

\global\long\def\sreg{\Sigma^{\text{reg}}}%

\global\long\def\sf#1{\mathrm{Sf}_{#1}}%

\global\long\def\mult#1#2{\mathrm{Mult}_{#1}\left(#2\right)}%

\global\long\def\L{L}%

\global\long\def\C{\mathbb{C}}%

\global\long\def\R{\mathbb{R}}%

\global\long\def\N{\mathbb{N}}%

\global\long\def\V{\mathcal{V}}%

\global\long\def\E{\mathcal{E}}%

\global\long\def\VR{\mathcal{V_{R}}}%

\global\long\def\Id{\mathbb{\boldsymbol{I}}}%

\global\long\def\Bell{\mathbb{\ensuremath{\boldsymbol{\ell}}}}%

\global\long\def\nmean#1{\langle#1\rangle_{n}}%

\global\long\def\hmean#1{\langle#1\rangle_{h}}%

\global\long\def\sreg{\Sigma^{\text{reg}}}%

\global\long\def\deltak{\langle\Delta k\rangle}%

\global\long\def\ncf{\mathcal{N}}%

\global\long\def\lmin{l_{min}}%

\global\long\def\Uv{\mathcal{U}_{v}}%

\global\long\def\Sv{\mathcal{S}_{v}}%

\global\long\def\tr{\text{tr}}%

\global\long\def\Ev{\mathcal{E}_{v}}%

\title{Differences between Robin and Neumann eigenvalues on metric graphs}
\author{Ram Band \and Holger Schanz \and Gilad Sofer}
\dedicatory{Dedicated to our teacher and mentor, Uzy Smilansky, on the occasion
of his 82-nd anniversary.}
\begin{abstract}
We consider the Laplacian on a metric graph, equipped with Robin ($\delta$-type)
vertex condition at some of the graph vertices and Neumann-Kirchhoff
condition at all others. The corresponding eigenvalues are called
Robin eigenvalues, whereas they are called Neumann eigenvalues if
the Neumann-Kirchhoff condition is imposed at all vertices. The sequence
of differences between these pairs of eigenvalues is called the Robin-Neumann
gap.

We prove that the limiting mean value of this sequence exists and
equals a geometric quantity, analogous to the one obtained for planar
domains \cite{RudWigYes_arxiv21}. Moreover, we show that the sequence
is uniformly bounded and provide explicit upper and lower bounds.
We also study the possible accumulation points of the sequence and
relate those to the associated probability distribution of the gaps.

To prove our main results, we prove a local Weyl law, as well as explicit
expressions for the second moments of the eigenfunction scattering
amplitudes.
\end{abstract}

\maketitle

\section{Introduction}

The differences between Robin and Neumann eigenvalues of the Laplacian
have been the focus of several recent works. Rudnick, Wigman, and
Yesha considered this sequence of Robin-Neumann gaps (RNG) for the
Laplacian on bounded planar domains and on the hemisphere \cite{RudWig_amq21,Rudnick2021,RudWigYes_arxiv21}.
They computed the limiting mean value of this RNG sequence, proved
some upper and lower bounds, and an almost sure convergence result.
Moreover, they posed stimulating open questions and interesting conjectures
-- such as the existence of planar domains with unbounded RNG sequence,
a lower bound for the RNG in the case of dimension larger than two,
and convergence in the case of a billiard with uniformly hyperbolic
dynamics.

The RNG sequence was studied by Rivi\`{e}re and Royer for the particular
case of metric star graphs in \cite{RivRoy_jphys20}, where they considered
a non self-adjoint Robin condition at the central vertex of a star
graph. They showed that the RNG sequence is bounded and with converging
mean value. They also expressed this limiting mean value in terms
of an associated probability distribution, and discussed some properties
of this distribution. While preparing the current manuscript for submission,
we became aware of the work \cite{BifKer_arXiv}, where Bifulco and
Kerner provide some generalization of the results above. In particular,
they prove a local Weyl law and use it to express the limiting mean
value of the RNG for Schr\"{o}dinger operators on metric graphs,
and extend some results for arbitrary self-adjoint vertex conditions.

In the current paper, we address the analogues of the results in \cite{RudWig_amq21,RudWigYes_arxiv21}
for the case of a metric (quantum) graph. We express the limiting
mean value (i.e., Ces\`{a}ro mean) of the RNG sequence (Theorem \ref{thm:mean_RNG}),
prove a local Weyl law, and express the second moments of the eigenfunction
scattering amplitudes (Theorem \ref{thm:Weyl-law}). We also provide
lower and upper bounds on the RNG (Theorem \ref{thm:explicit_bounds}),
present the associated probability measure (Theorem \ref{thm:probability}),
and use it to study the convergence of subsequences (Theorem \ref{thm:converging_subsequence}).
In doing so, we attempt to answer some of the questions proposed in
previous works, and compare the results to the ones obtained for domains
and star graphs (see Section \ref{sec:discussion}).

To end this introductory part, we note that the dependence of the
spectrum on the boundary conditions in planar domains has long been
a topic of interest in physics. For instance, the group of Uzy Smilansky
used variations of the boundary conditions as a tool in the study
of Gutzwiller's trace formula for quantum billiards \cite{SiePriSmiUssSch_1995}.
Our teacher Uzy has also inspired the present work, and we would therefore
like to dedicate it to his anniversary.

\subsection{Basic definitions and notations \label{subsec:Basic-definitions}}

We consider a metric graph $\Gamma$, with $\V$ and $\E$ being its
vertex set and edge set respectively. Denoting $E:=\left|\E\right|$,
the edge lengths of $\Gamma$ are determined by the vector $\vec{\ell}\in\mathbb{R}_{+}^{E}$
of positive entries. Each edge $e\in\mathcal{E}$ is identified with
the interval $\left[0,\ell_{e}\right]$, so that under the natural
identification of vertices connected to the appropriate edges, $\Gamma$
is a compact metric space. The total length of the graph is denoted
by $\left|\Gamma\right|:=\sum_{e\in\E}\ell_{e}$. For each vertex
$v\in\mathcal{V}$, we denote the set of edges connected to $v$ by
$\mathcal{E}_{v}$, and moreover denote $\deg\left(v\right):=\left|\mathcal{E}_{v}\right|$.

Given a metric graph $\Gamma$, we consider the Hilbert space $L^{2}\left(\Gamma\right):=\oplus_{e\in\mathcal{E}}L^{2}\left(\left[0,\ell_{e}\right]\right)$.
We can then define the Neumann-Kirchhoff Laplacian (also known as
the \emph{standard} Laplacian) by $H^{(0)}=-\frac{d^{2}}{dx^{2}}$
acting on each edge, with domain consisting of all Sobolev functions
$f\in W^{1,2}\left(\Gamma\right)$ which satisfy for all $v\in\mathcal{V}$:
\begin{align}
 & \text{Continuity: }\forall e,e'\in\mathcal{E}_{v},\quad\ensuremath{f|_{e}\left(v\right)=f|_{e'}\left(v\right)},\label{eq:-15-1}\\
 & \text{Current conservation: \ensuremath{\sum_{e\in\Ev}f'|_{e}\left(v\right)=0,}}\label{eq:-16-1}
\end{align}
where by convention, all derivatives are taken in the outward direction
from the vertex. The operator $H^{(0)}$ is self-adjoint, and its
spectrum is infinite, discrete, and bounded from below (\cite{BerKuc_graphs}).
We thus denote the spectrum of $H^{(0)}$ by $\lambda_{1}\leq\lambda_{2}\leq...\nearrow\infty$,
with a complete orthonormal set of real eigenfunctions $f_{1},f_{2},...$.

\bigskip

Let $\Gamma$ be a metric graph, initially endowed with the Neumann-Kirchhoff
Laplacian, as described above. We introduce a perturbation to our
initial operator $H^{(0)}$ by selecting a finite subset of vertices
$\VR\subset\mathcal{V}$, and on this subset of vertices imposing
the Robin vertex condition (also known as $\delta$-type vertex condition)
with parameter $\sigma\geq0$:
\begin{align}
 & \text{Continuity: \ensuremath{\forall e,e'\in\mathcal{E}_{v},\quad f|_{e}\left(v\right)=f|_{e'}\left(v\right)=:f\left(v\right),}}\label{eq:-17-1}\\
 & \text{\text{Robin condition:} \ensuremath{\sum_{e\in\Ev}f'|_{e}\left(v\right)=\sigma f\left(v\right)}},\label{eq:-18-1}
\end{align}
for all $v\in\VR$. The case $\sigma=0$ corresponds to the Neumann-Kirchhoff
condition; namely, no perturbation at all. Further note that condition
(\ref{eq:-18-1}) is the analogue of the Robin boundary condition
for manifolds (hence its name).
\begin{rem}
It is also possible to impose the Robin condition at an interior point
of an edge. To do so, one simply declares such an interior point as
a degree two vertex in $\mathcal{V}$ and adds this vertex to $\VR$.
\end{rem}

We denote the new operator by $H^{(\sigma)}$, and its eigenvalues
by $\left(\lambda_{n}\left(\sigma\right)\right)_{n=1}^{\infty}$.
We may also refer to the square roots of the eigenvalues (a.k.a wave
numbers), $\left(k_{n}\left(\sigma\right)\right)_{n=1}^{\infty}:=\left(\sqrt{\lambda_{n}\left(\sigma\right)}\right)_{n=1}^{\infty}$,
which are well-defined and non-negative since $\spec{H^{(\sigma)}}\subset[0,\infty)$
for $\sigma\geq0$. It is known that the eigenvalues of $H^{(\sigma)}$
are non-decreasing with respect to $\sigma$, see \cite[prop. 3.1.6]{BerKuc_graphs}.
To quantify this increase, define the \emph{Robin-Neumann gaps} (RNG)
by
\begin{equation}
d_{n}\left(\sigma\right):=\lambda_{n}\left(\sigma\right)-\lambda_{n}\left(0\right).\label{eq:-21-1}
\end{equation}
This gives an infinite sequence of functions $\left(d_{n}\left(\sigma\right)\right)_{n=1}^{\infty}$,
which measures the increase in the spectrum of $H^{(\sigma)}$ due
to the $\delta$ perturbation (see Figure \ref{fig:RNG-demo}). The
current paper focuses on studying the main properties of this sequence. 

\begin{figure}
\centerline{ \includegraphics[width=0.5\textwidth]{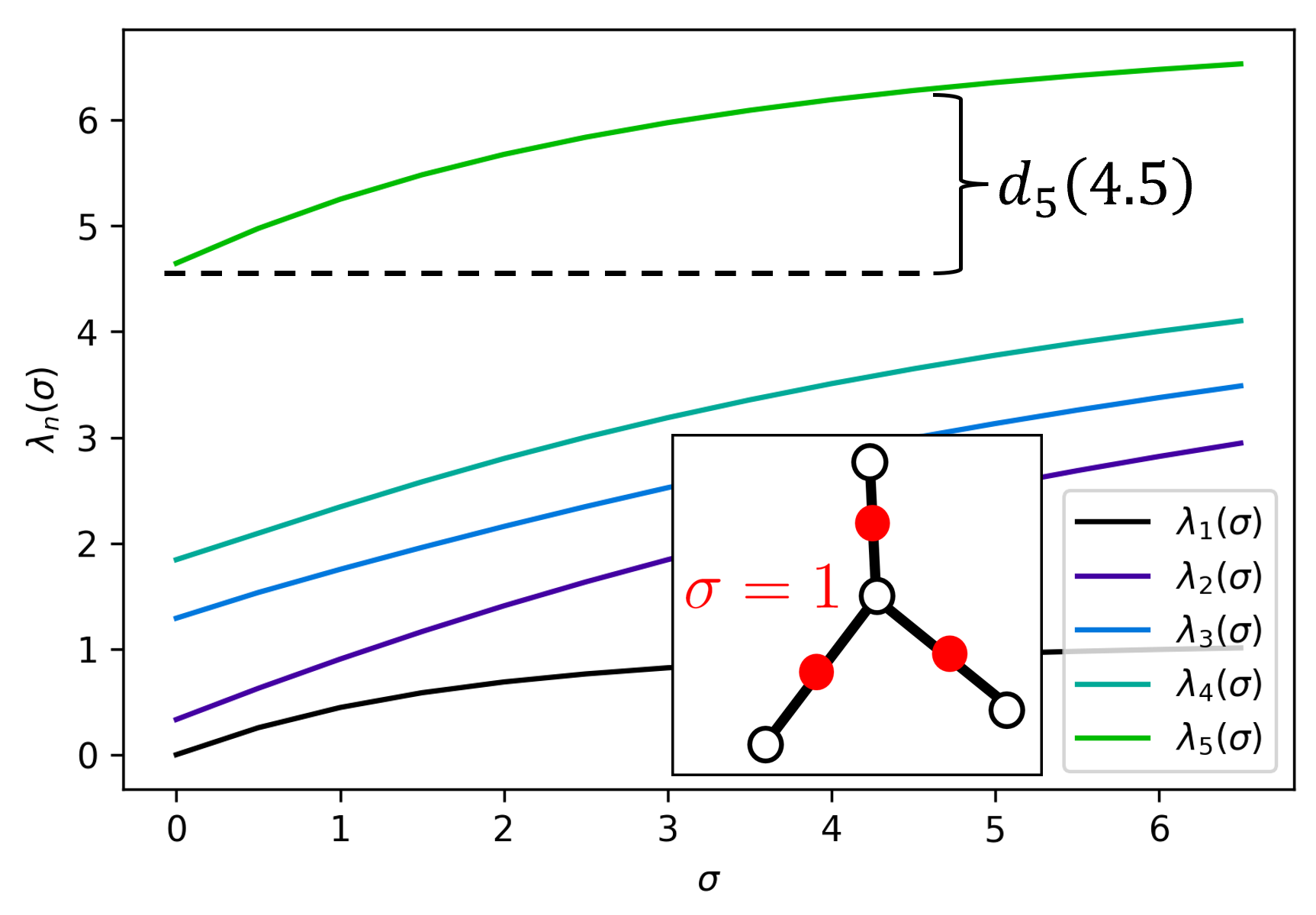}
} \caption{\label{fig:RNG-demo} The Robin eigenvalues $\lambda_{n}\left(\sigma\right)$
for a star graph, along with the Robin-Neumann gap $d_{5}\left(4.5\right)$.
The Robin vertices are marked in red.}
\end{figure}
\begin{figure}
\centerline{ \includegraphics[width=0.7\textwidth]{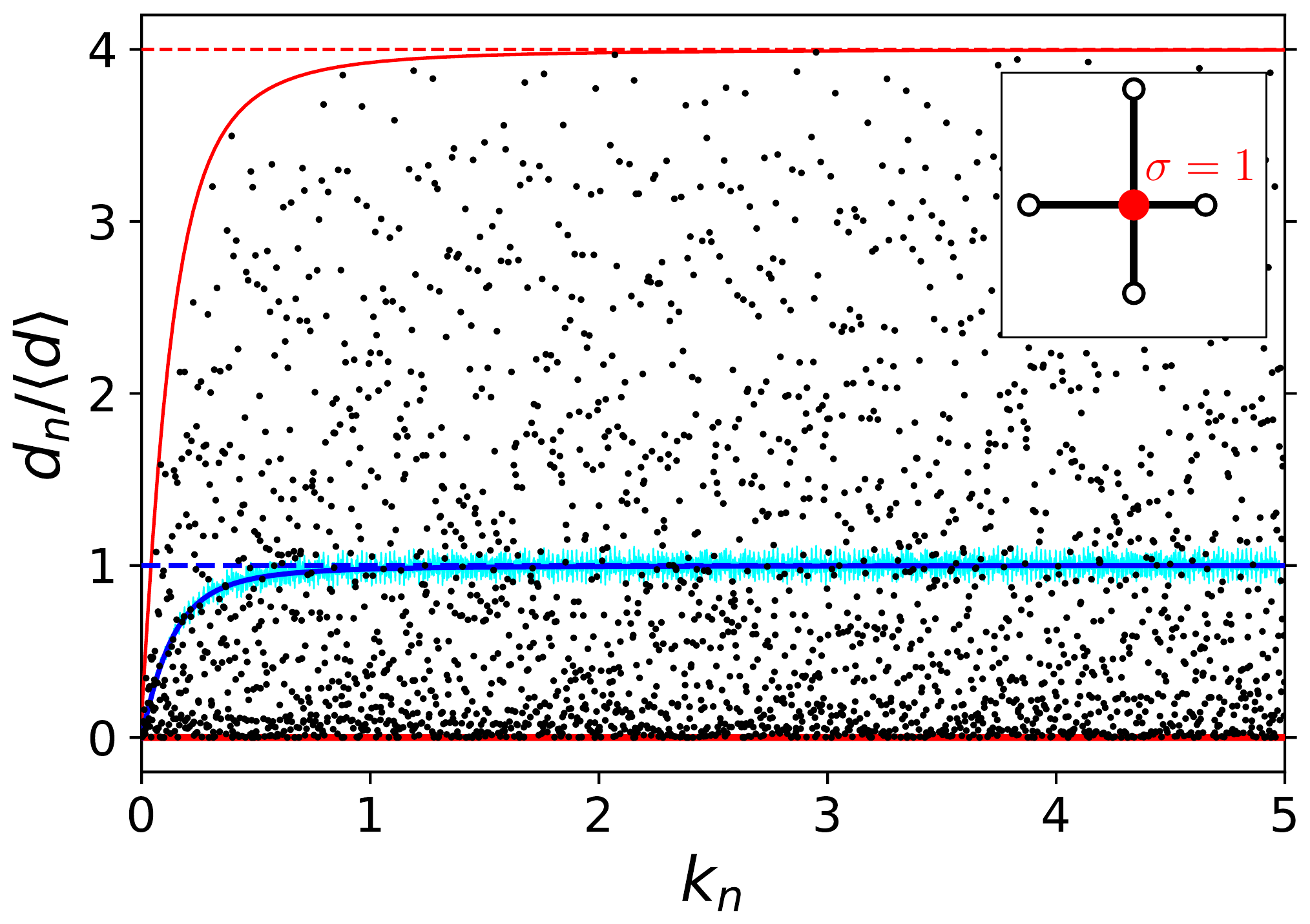} }

\caption{\label{fig:star_gap} Scatter plot of the first $2,500$ Robin-Neumann
gaps for a star graph with four edges and Robin condition at the central
vertex, normalized so that $\left\langle d\right\rangle _{n}\left(\sigma\right)=1$.
The light blue line is a running average and the blue lines on top
of it are the analytic results from Equations (\ref{eq:RNG_mean_const}),
(\ref{eq:RNG_mean_arctan}). The red dashed line is the first upper
bound presented in Equation (\ref{eq:explicit_bounds}), while the
solid red line is the finer upper bound appearing in Equation (\ref{eq:upperbound})
(under the star decomposition described in Subsection \ref{subsec:Explicit-estimate-of}).}
\end{figure}

\subsection{Main results}
\begin{defn}
\label{def:Cesaro}Given a sequence of numbers $\left(c_{n}\right)_{n=1}^{\infty}$,
the Ces\`{a}ro mean (or Ces\`{a}ro sum) of the sequence is 
\begin{equation}
\left\langle c\right\rangle _{n}:=\lim_{N\rightarrow\infty}\frac{1}{N}\sum_{n=1}^{N}c_{n},\label{eq:-40}
\end{equation}
\end{defn}

assuming that the limit exists.

The first result concerns the Ces\`{a}ro mean of the Robin-Neumann
gaps:
\begin{thm}
\textcolor{purple}{\label{thm:mean_RNG}} The Ces\`{a}ro mean of
the Robin-Neumann gap exists and satisfies
\begin{equation}
\left\langle d\right\rangle _{n}\left(\sigma\right)=\frac{2\sigma}{\left|\Gamma\right|}\sum_{v\in\V_{\mathcal{R}}}\frac{1}{\deg\left(v\right)},\label{eq:RNG_mean_const}
\end{equation}
where $\left|\Gamma\right|$ is the total length of the graph and
$\deg\left(v\right)$ is the degree of the vertex $v$.
\end{thm}

To prove Theorem \ref{thm:mean_RNG}, we prove a result which has
its own interest -- a local Weyl law and expressions for the second
moments of the eigenfunction scattering amplitudes:
\begin{thm}
\label{thm:Weyl-law}Denote by $\left(f_{n}\right)_{n=1}^{\infty}$
the $L^{2}$ normalized eigenfunctions of the Neumann-Kirchhoff Laplacian
$H^{(0)}$. Then for each vertex $v\in\mathcal{V}$,
\begin{equation}
\left\langle \left|f\left(v\right)\right|^{2}\right\rangle _{n}=\frac{2}{\deg\left(v\right)\left|\Gamma\right|}.\label{eq:-15}
\end{equation}
Moreover, expressing these eigenfunctions on each edge $e\in\E$ by
\begin{equation}
f_{n}|_{e}\left(x\right)=\left(a_{e}\right)_{n}\thinspace\exp\left(\rmi k_{n}x\right)+\left(a_{\hat{e}}\right)_{n}\thinspace\exp\left(\rmi k_{n}\left(\ell_{e}-x\right)\right),\label{eq:-45}
\end{equation}
we have
\begin{align}
 & \forall e\in\mathcal{E},\quad\left\langle \left|a_{e}\right|^{2}\right\rangle _{n}=\left\langle \left|a_{\hat{e}}\right|^{2}\right\rangle _{n}=\frac{1}{2\left|\Gamma\right|},\label{eq:-16}\\
 & \forall e_{1},e_{2}\in\mathcal{E},\textrm{ such that \ensuremath{e_{1}\neq e_{2},\quad}}\left\langle a_{e_{1}}\overline{a_{e_{2}}}\right\rangle _{n}=0.\label{eq:-17}
\end{align}
\end{thm}

\begin{rem}
The local Weyl law for quantum graphs was recently proven by Borthwick,
Harrell and Jones in \cite{Borthwick2022} via heat kernel methods.
Afterwards, Bifulco and Kerner had shown that the result may be generalized
to Schr\"{o}dinger operators with arbitrary self-adjoint vertex conditions
\cite{BifKer_arXiv}.
\end{rem}

\bigskip
\begin{thm}
\label{thm:Lipschitz_RNG}~ The sequence of functions $\left(d_{n}\left(\sigma\right)\right)_{n=1}^{\infty}$
is uniformly Lipschitz continuous in $[0,\infty)$. Namely, there
exists $C\in\R$ such that 
\begin{equation}
\forall n\in\N,~~\forall\sigma_{1},\sigma_{2}\geq0,\quad\left|d_{n}\left(\sigma_{1}\right)-d_{n}\left(\sigma_{2}\right)\right|\leq C\left|\sigma_{1}-\sigma_{2}\right|.\label{eq:Lipschitz_RNG}
\end{equation}
In particular, the sequence of functions $\left(d_{n}\left(\sigma\right)\right)_{n=1}^{\infty}$
is uniformly bounded on any compact interval.
\end{thm}

\bigskip

Theorem \ref{thm:Lipschitz_RNG} has a very short proof which appears
in Section \ref{sec:Lipschitz-proof}. With some more effort, it is
possible to obtain explicit bounds for the Robin-Neumann gap. In order
to do so, we introduce an auxiliary construction.

Let $\Gamma$ be a metric graph. A \emph{star decomposition} of $\Gamma$
is a partition of $\Gamma$ into star graphs, whose central vertices
are the vertices of $\Gamma$. Such a partition may be described by
introducing an auxiliary vertex $u_{e}$ along each edge $e\in\E$.
The vertex $u_{e}$ may be positioned in the interior of $e$ (so
that $\deg(u_{e})=2$), or at its boundary (so that $u_{e}\in\V$).
We denote the set of all such auxiliary vertices connected to $v$
by $\Uv$. Hence, for each $v\in\V$, the associated star subgraph
of the partition consists of the central vertex $v$ and all vertices
in $\Uv$. This star subgraph is denoted by $\Sv$, and the edge lengths
of this graph are denoted by $\left\{ s_{v,u}\right\} _{u\in\Uv}$.
For example, if an edge $e$ of $\Gamma$ connects the vertices $v,w\in\V$
and $u_{e}\in\Uv\cap\mathcal{U}_{v'}$~~, then $\ell_{e}=s_{v,u_{e}}+s_{w,u_{e}}$.
Note that it is also possible to have $s_{v,u_{e}}=0$ if the corresponding
auxiliary vertex is placed at the boundary of the edge $e$, such
that $u_{e}=v$. We further denote the total edge length of each star
by 
\begin{equation}
\left|\Sv\right|:=\sum_{u\in\Uv}s_{v,u}.\label{-2}
\end{equation}
An example of a star decomposition for a tetrahedron graph can be
seen in Figure \ref{fig:star_decomp}.

\begin{figure}
\centerline{ \includegraphics[width=0.5\textwidth]{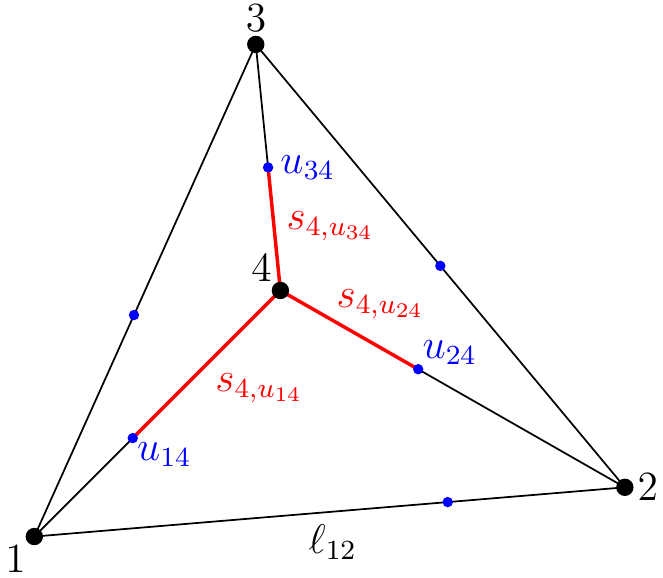}
} \caption{\label{fig:star_decomp} Star decomposition of a tetrahedron graph.
The vertices and edges of the original graph are shown in black. The
small blue dots correspond to the auxiliary vertices. The star graph
around vertex 4 is highlighted in red.}
\end{figure}

\begin{thm}
\label{thm:explicit_bounds} For any star decomposition of the metric
graph $\Gamma$ and $\sigma>0$: 
\begin{eqnarray}
0\leq & d_{n}\left(\sigma\right)<\frac{2\sigma}{\min_{v\in\VR}\left|\Sv\right|}.\label{eq:explicit_bounds}
\end{eqnarray}
In particular,
\begin{equation}
d_{n}\left(\sigma\right)<\frac{4\sigma}{\ell_{\min}},\label{eq:-13}
\end{equation}
where $\ell_{\min}$ is the length of the shortest edge in $\Gamma$.
\end{thm}

In fact, Proposition \ref{prop:integral_estimate} which appears in
the sequel gives a better upper bound than (\ref{eq:explicit_bounds}),
but is expressed in a more cumbersome manner. Figure \ref{fig:star_gap}
demonstrates both the bound (\ref{eq:explicit_bounds}) above (dashed
red curve) and the better bound (\ref{eq:upperbound}) in Proposition
\ref{prop:integral_estimate} (solid red curve). Note that for a star
graph with a single Robin vertex $v$ at its center, the optimal bound
in (\ref{eq:explicit_bounds}) is obtained by the star partition which
consists of just a single star -- the whole graph. For this partition,
$\left|\Sv\right|=$$\left|\Gamma\right|$ and the upper bound in
(\ref{eq:explicit_bounds}) is $\frac{2\sigma}{\left|\Gamma\right|}$.
See further discussion in Subsection \ref{subsec:Explicit-estimate-of}.

\bigskip

After presenting results about the mean value and bounds of the RNG,
we now turn to discuss properties of its value distribution and limit
points of its subsequences.
\begin{thm}
\label{thm:probability}For $\sigma>0$ define the function
\begin{align}
 & F_{\sigma}:\mathbb{R}\rightarrow\left[0,1\right],\label{eq:-20}\\
 & F_{\sigma}\left(x\right):=\lim_{N\rightarrow\infty}\frac{1}{N}\left|\left\{ n\leq N~:~d_{n}\left(\sigma\right)\leq x\right\} \right|.
\end{align}
Then $F_{\sigma}$ is a cumulative distribution function whose associated
probability measure $\mu_{\sigma}$ is compactly supported on $\left[0,\frac{4\sigma}{\ell_{\min}}\right]$.

If for all $e,e'\in\E$, $\ell_{e}/\ell_{e'}\in\mathbb{Q}$, then
$\mu_{\sigma}$ is finitely supported.
\end{thm}

\begin{rem}
We conjecture that if not all ratios of edge lengths are rational,
then the measure $\mu_{\sigma}$ is absolutely continuous with respect
to the Lebesgue measure. This is further discussed in Remark \ref{rem:Absolutely-continuous-conjecture}
after the proof of the theorem.
\end{rem}

\begin{rem}
Theorem \ref{thm:probability} is similar in spirit to theorem $1.3$
in \cite{RivRoy_jphys20}, which holds for star graphs with complex
Robin parameter.
\end{rem}

Interestingly, applying Theorem \ref{thm:probability} yields information
about the possible limit points of the RNG sequence. While for particular
values of $\sigma$ the sequence $\left(d_{n}\left(\sigma\right)\right)_{n=1}^{\infty}$
does not necessarily converge (as seen in Figure \ref{fig:star_gap}),
we show that there exist subsequences of $\left(d_{n}\left(\sigma\right)\right)_{n=1}^{\infty}$
which converge uniformly to a linear function.
\begin{thm}
\label{thm:converging_subsequence} There exists a finite collection
of closed intervals $\left(\left[a_{i},b_{i}\right]\right)_{i=1}^{N}$,
such that for every $c\in\bigcup_{i=1}^{N}\left[a_{i},b_{i}\right]$,
there exists a subsequence $\left(d_{n_{m}}\left(\sigma\right)\right)_{m=1}^{\infty}$
which converges to the linear function $c\sigma$ uniformly on any
compact interval. These are the only possible partial limits of $d_{n}\left(\sigma\right)$.\\
The intervals $\left(\left[a_{i},b_{i}\right]\right)_{i=1}^{N}$ do
not depend on the graph edge lengths, as long as the edge lengths
are linearly independent over $\mathbb{Q}$.\\
In the other extreme, if for all $e,e'\in\E,$ $\ell_{e}/\ell_{e'}\in\mathbb{Q}$
then these intervals are in fact degenerate ($a_{i}=b_{i}$).
\end{thm}

\bigskip
\begin{rem}
We comment on several possible generalizations of the results above.
\begin{enumerate}
\item For $\sigma<0$, the operator $H_{\sigma}$ has only finitely many
negative eigenvalues. Thus, all results except for Theorem \ref{thm:explicit_bounds}
still hold for $\sigma<0$ (although the corresponding proofs are
slightly more subtle). In addition, the results may be naturally extended
for the case where each vertex in $\VR$ has its own value of $\sigma$
(rather than $\sigma$ being the same at all these vertices). The
adaption of the statements is rather straightforward.
\item In all results expect for Theorem \ref{thm:explicit_bounds}, one
may replace the Laplacian with a Schr\"{o}dinger operator $H^{(\sigma)}=-\frac{d^{2}}{dx^{2}}+V\left(x\right)$,
with $V\left(x\right)\in L^{\infty}\left(\Gamma\right)$ and the same
vertex conditions. This is done in \cite{BifKer_arXiv} for Theorems
\ref{thm:mean_RNG}, \ref{thm:Weyl-law}.
\end{enumerate}
\end{rem}

The structure of the paper is as follows. Section \ref{sec:Review-tools}
is devoted to introducing several tools which are required for the
proofs. Section \ref{sec:Lipschitz-proof} contains the proof of Theorem
\ref{thm:Lipschitz_RNG}. The local Weyl law and the expressions of
second moments of the scattering amplitudes (Theorem \ref{thm:Weyl-law})
are proven in Section \ref{sec:Proof-Weyl}, and are then used in
Section \ref{sec:mean-proof} to compute the Ces\`{a}ro mean which
is stated in Theorem \ref{thm:mean_RNG}. Section \ref{sec:conv-sub-proof}
is dedicated to the proofs of Theorems \ref{thm:probability} and
\ref{thm:converging_subsequence}; it also contains a discussion of
the probability distribution and limit points of the RNG. The explicit
bounds on the RNG from Theorem \ref{thm:explicit_bounds} are proven
and discussed in Section \ref{sec:bounds proof}. Section \ref{sec:discuss_mean_RNG}
provides an additional approach for deriving the mean value of the
RNG and its convergence to the limiting mean value. Finally, Section
\ref{sec:discussion} contains a comparison of our results to the
analogous statements for domains and star graphs, as well as several
open questions.

\bigskip

\section{Review of tools and methods for proofs\label{sec:Review-tools}}

In this section, we review some existing tools of spectral analysis
on metric graphs, which are useful in the proofs presented in the
following sections.

\subsection{Expressing the Robin-Neumann gap via eigenfunction values}

A main ingredient used in the proofs is the following formula, which
relates the RNG to the values that the eigenfunctions attain on the
set $\mathcal{V}_{R}$.
\begin{lem}
\label{lem:hadamard}The RNG is given by
\begin{equation}
d_{n}\left(\sigma\right)=\sum_{v\in\mathcal{V}_{\mathcal{R}}}\int_{0}^{\sigma}\left|f_{n}^{\left(t\right)}\left(v\right)\right|^{2}dt,\label{eq:-1}
\end{equation}
where $f_{n}^{\left(t\right)}$ is an $n^{\textrm{th}}$ $L^{2}$
normalized eigenfunction of $H^{(t)}$.
\end{lem}

\begin{rem}
\label{rem:additivity}Lemma \ref{lem:hadamard} allows us to assume
in the forthcoming proofs that the graph contains a single Robin vertex
$v$. The proof for more than one Robin vertex would then follow from
the additivity of Equation (\ref{eq:-1}). 
\end{rem}

\begin{proof}
We use a straightforward generalization\footnote{The original formula refers only to a single Robin vertex. Using additivity
arguments in its proof, the more generalized version with multiple
Robin vertices follows. } of a formula from \cite[prop. 3.1.6]{BerKuc_incol12},
\begin{equation}
\frac{d\lambda_{n}\left(t\right)}{dt}=\sum_{v\in\mathcal{V}_{\mathcal{R}}}\left|f_{n}^{\left(t\right)}\left(v\right)\right|^{2}.\label{eq:-43}
\end{equation}
The formula holds for all values of $t$ for which $\lambda_{n}\left(t\right)$
is simple.

Note that unless $\lambda_{n}\left(t\right)$ is a multiple eigenvalue
for all $t\geq0$, then the set $D\subset\left[0,\sigma\right]$ of
$t$ values for which $\lambda_{n}\left(t\right)$ is non-simple must
be finite. Indeed, by \cite[thm 3.1.2]{BerKuc_graphs}, the eigenvalues
$\lambda_{n}\left(t\right)$ are piecewise real analytic in $t$ (see
also \cite[thms. 3.8, 3.10]{BerKuc_incol12}). Thus, if $D\subset\left[0,\sigma\right]$
was infinite, two of the eigenvalue curves would agree on a set with
an accumulation point, and thus agree everywhere. Furthermore, if
$\lambda_{n}\left(t\right)$ is a multiple eigenvalue for all $t$,
then by \cite[thm 3.1.4]{BerKuc_graphs}, one can locally choose an
analytic orthonormal basis for the eigenspace of $\lambda_{n}\left(t\right)$.
After doing so, \cite[thms 3.23, 3.24]{Latushkin2020} shows that
(\ref{eq:-43}) also holds for the case where $\lambda_{n}\left(t\right)$
is a multiple eigenvalue for all $t$, where this time $f_{n}^{\left(t\right)}$
is chosen to be an arbitrary $L^{2}$ normalized eigenfunction from
the given eigenspace.

Either way, we conclude that (\ref{eq:-43}) holds almost everywhere,
and thus
\begin{equation}
d_{n}\left(\sigma\right)=\lambda_{n}\left(\sigma\right)-\lambda_{n}\left(0\right)=\int_{0}^{\sigma}\frac{d\lambda_{n}\left(t\right)}{dt}dt=\sum_{v\in\mathcal{V}_{\mathcal{R}}}\int_{0}^{\sigma}\left|f_{n}^{\left(t\right)}\left(v\right)\right|^{2}dt.\label{eq:-44}
\end{equation}
\end{proof}
\begin{rem}
When applying Lemma \ref{lem:hadamard} in the subsequent sections,
we shall conveniently assume that $\lambda_{n}\left(t\right)$ is
not a multiple eigenvalue for all $t$. The proof above shows that
this assumption makes no difference, as long as one appropriately
chooses the eigenfunction $f_{n}^{\left(t\right)}$ in the degenerate
case. To avoid this technicality, we focus on the more standard non-degenerate
case.
\end{rem}

\subsection{Scattering formalism and the secular equation}

Let $H^{(\sigma)}$ be the operator introduced in Subsection \ref{subsec:Basic-definitions}.
An eigenfunction $f$ of $H^{(\sigma)}$, with eigenvalue $k^{2}>0$,
can be written on each graph edge $e\in\mathcal{E}$ as
\begin{equation}
f|_{e}\left(x\right)=a_{e}\exp\left(\rmi kx\right)+a_{\hat{e}}\exp\left(\rmi k\left(\ell_{e}-x\right)\right).\label{eq:-6}
\end{equation}
Thus, $f$ can be described by the vector of coefficients $\vec{a}=\left(a_{1},a_{\hat{1}},...,a_{E},a_{\hat{E}}\right)\in\mathbb{C}^{2E}$,
which depends on the wave number $k$. We can think of $a_{e}$ as
representing the amplitude of a (one-dimensional) plane wave propagating
along the edge $e$ in a certain direction. Similarly, $a_{\hat{e}}$
represents a wave propagating along the same edge, but in the opposite
direction. Adopting this physical interpretation, we may consider
the graph as a directed graph. Hence, we denote directed edges with
opposite directions by $e$ and $\hat{e}$. We follow here the theory
which was originally developed by Kottos and Smilanksy in \cite{KotSmi_prl00,KotSmi_ap99}.

We note that (\ref{eq:-6}) yields straightforwardly that $-f''=k^{2}f$.
Yet, the form (\ref{eq:-6}) does not guarantee that $f$ satisfies
the Robin vertex conditions, as in (\ref{eq:-17-1}),(\ref{eq:-18-1}).
In order to satisfy this, one introduces a diagonal edge length matrix
$\L:=\mathrm{diag}\left(\ell_{1},\ell_{1},\ell_{2},\ell_{2},...,\ell_{E},\ell_{E}\right)$
and a unitary matrix $S^{(\sigma)}\in U\left(2E\right)$. The expression
for $S^{\left(\sigma\right)}$ depends on the vertex conditions and
the wave number $k$ as follows. Given two directed edges $e,e'$,
we write that $e\rightarrow e'$ at $v$ if the end vertex of $e$
is $v$ and the starting vertex of $e'$ is $v$. Using this notation,
we set (see \cite[eq. (3.6)]{GnuSmi_ap06})
\begin{equation}
S_{e'e}^{\left(\sigma\right)}\left(k\right)=\begin{cases}
\frac{2}{\deg\left(v\right)+\frac{\rmi\sigma}{k}}-1 & e'=\hat{e}\\
\frac{2}{\deg\left(v\right)+\frac{\rmi\sigma}{k}} & e\rightarrow e'\text{ at \ensuremath{v} and }e'\neq\hat{e}\\
0 & \text{Otherwise.}
\end{cases}\label{eq:-32-2}
\end{equation}
After some computation, one gets that

\begin{equation}
\left(I-S^{(\sigma)}\rme^{\rmi k\L}\right)\vec{a}=0,\label{eq:-2}
\end{equation}
whenever $k^{2}>0$ is an eigenvalue of the graph. In such a case,
the amplitude vector $\vec{a}$ in (\ref{eq:-2}) characterizes the
eigenfunction $f$ of that eigenvalue, as in (\ref{eq:-6}). This
formalism also bears an interesting physical point of view in terms
of scattering dynamics on the graph -- see \cite{GnuSmi_ap06} for
a more elaborate description. We summarize the discussion above with
\begin{thm}
\label{thm:secular-equation}(\cite{KotSmi_ap99,KotSmi_prl00}) Let
$\sigma\geq0$. $k^{2}>0$ is an eigenvalue of $H^{(\sigma)}$ if
and only if
\begin{equation}
\det\left(I-S^{(\sigma)}\rme^{\rmi k\L}\right)=0.\label{eq:sec-equation}
\end{equation}
\end{thm}

The equation $\det\left(I-S\rme^{\rmi k\L}\right)=0$ which describes
the graph's eigenvalues is frequently called the \emph{secular equation}
(or the \emph{secular function}, when referring just to its left hand
side). Theorem \ref{thm:secular-equation} actually holds for a more
general class of vertex conditions (once $S$ is appropriately modified),
see \cite{BerKuc_graphs}.

\subsection{The secular manifold \label{subsec:The-secular-manifold}}

We now focus on the Neumann-Kirchhoff Laplacian $H^{(0)}$. Motivated
by Theorem \ref{thm:secular-equation}, we can define the following
function on $\mathbb{R}^{E}$:
\begin{equation}
\tilde{F}\left(\vec{x}\right)=\det\left(I-S\rme^{\rmi\boldsymbol{x}}\right),\label{eq:-3}
\end{equation}
where $\vec{x}:=(x_{1},...,x_{E})$, $\boldsymbol{x}:=\mathrm{diag}\left(x_{1},x_{1},...,x_{E},x_{E}\right)$,
and $S$ here is taken to be $S^{\left(\sigma\right)}$ with $\sigma=0$
(see (\ref{eq:-32-2})). The function $\tilde{F}$ is clearly $2\pi$-periodic
in each of its components, and we can thus consider it as a function
$F$ on the torus $\mathbb{T}^{E}=\mathbb{R}^{E}/2\pi\mathbb{Z}^{E}$.

By Theorem \ref{thm:secular-equation}, we conclude that $k^{2}$
is an eigenvalue of $H^{\left(0\right)}$ if and only if $F\left(\vec{\kappa}\right)=0$,
where $\vec{\kappa}:=k\vec{\ell}\mod2\pi=\left(k\ell_{1}\mod2\pi,...,k\ell_{E}\mod2\pi\right)$.
With this in mind, we can define the set
\begin{equation}
\Sigma:=\left\{ \vec{\kappa}\in\mathbb{T}^{E}:F\left(\vec{\kappa}\right)=0\right\} .\label{eq:-38-2}
\end{equation}
$\Sigma$ is a compact algebraic subvariety of the torus $\mathbb{T}^{E}$
known as the secular manifold\footnote{This is a slight misnomer, since generally the secular manifold could
have singular points.}, see \cite{AloBanBer_cmp18}. The eigenvalues of $H^{(0)}$ are thus
determined by the $k$ values such that the linear torus flow $\phi\left(k\right)=\left(k\ell_{1}\mod2\pi,...,k\ell_{E}\mod2\pi\right)$
intersects the secular manifold\footnote{To be precise, those $k$ values are the square roots of the eigenvalues
of $H^{(0)}$.}. Moreover, we can define a map between an eigenfunction $f$ of $H^{(0)}$
with eigenvalue $k^{2}$ and the corresponding amplitude vector, $\vec{a}\in\ker\left(I-S\rme^{\rmi k\L}\right)$,
such that the relation (\ref{eq:-6}) holds. This map is in fact a
linear bijection between the $k^{2}$-eigenspace of $H^{(0)}$ and
$\ker\left(I-S\rme^{\rmi k\L}\right)$, \cite[lem. 4.12]{Alon_PhDThesis}.
With this in mind, a special emphasis is given to the simple eigenvalues
of $H^{(0)}$. To treat those, we define: 
\begin{equation}
\sreg:=\left\{ \vec{\kappa}\in\Sigma:\dim\ker\left(I-S\rme^{\rmi\boldsymbol{\kappa}}\right)=1\right\} ,\label{eq:sreg}
\end{equation}
where $\boldsymbol{\kappa}:=\mathrm{diag}\left(\kappa_{1},\kappa_{1},...,\kappa_{E},\kappa_{E}\right)$.
$\sreg$ is a smooth submanifold of the torus of codimension one (\cite[thm 1.1]{CdV_ahp15},\cite[thm 3.6]{AloBanBer_cmp18}).
Furthermore, the bijection mentioned above implies that $\sreg$ classifies
the simple eigenvalues of $H^{(0)}$. Namely, $k\vec{\ell}\mod2\pi\in\sreg$
if and only if $k^{2}$ is a\emph{ }\emph{simple} eigenvalue of $H^{(0)}$.
For such $\vec{\kappa}:=k\vec{\ell}\mod2\pi\in\sreg$, taking $\vec{a}(\vec{\kappa})\in\ker\left(I-S\rme^{\rmi\boldsymbol{\kappa}}\right)$
provides through (\ref{eq:-6}) the unique eigenfunction $f$ (up
to scalar multiplication) whose eigenvalue is $k^{2}$. Since we know
that $f$ is real (up to scalar multiplication), we get 
\begin{align}
\forall e\in\E,~~\forall x\in\left[0,\ell_{e}\right],\quad\quad f|_{e}\left(x\right) & =\overline{f|_{e}\left(x\right)}\quad\Leftrightarrow\label{eq:inversion-1}\\
a_{e}\exp\left(\rmi kx\right)+a_{\hat{e}}\exp\left(\rmi k\left(\ell_{e}-x\right)\right)= & \overline{a_{e}}\exp\left(-\rmi kx\right)+\overline{a_{\hat{e}}}\exp\left(-\rmi k\left(\ell_{e}-x\right)\right),\label{eq:inversion-2}
\end{align}
and so 
\begin{align}
a_{e} & =\overline{a_{\hat{e}}}\exp\left(-\rmi k\ell_{e}\right),\label{eq:inversion-3}\\
a_{\hat{e}} & =\overline{a_{e}}\exp\left(-\rmi k\ell_{e}\right).\label{eq:inversion-4}
\end{align}
Keeping in mind that the coefficients $a_{e}$ depend on $\vec{\kappa}$
(a dependence which we omitted for brevity), we get 
\begin{align}
a_{e}(\vec{\kappa}) & =\overline{a_{\hat{e}}(\vec{\kappa})}\exp\left(-\rmi\vec{\kappa}_{e}\right),\label{eq: amplitude_relation_1-1}\\
a_{\hat{e}}(\vec{\kappa}) & =\overline{a_{e}(\vec{\kappa})}\exp\left(-\rmi\vec{\kappa}_{e}\right).\label{eq: amplitude_relation_2-1}
\end{align}

We argue that (\ref{eq: amplitude_relation_1-1}),(\ref{eq: amplitude_relation_2-1})
hold for all $\vec{\kappa}\in\sreg$, as is explained in the following.
Above, (\ref{eq: amplitude_relation_1-1}),(\ref{eq: amplitude_relation_2-1})
were derived only for those $\vec{\kappa}\in\sreg$ values for which
$\vec{\kappa}:=k\vec{\ell}\mod2\pi$ holds for some value of $k\in\R$.
These values form only a countable subset of points in $\sreg$. Nevertheless,
all points of $\sreg$ have a similar spectral meaning. Namely, if
$\vec{\kappa}\in\sreg$ is such that there is no $k\in\R$ satisfying
$\vec{\kappa}:=k\vec{\ell}\mod2\pi$, we simply pick a different $\vec{\ell}'\in\R^{E}$
such that $\vec{\kappa}=k'\vec{\ell}'\mod2\pi$ for some $k'\in\R$.
Doing so means that we consider a graph $\Gamma'$ with the same connectivity
as $\Gamma$, but with edge lengths given by $\vec{\ell}'$. We obtain
that $(k')^{2}$ is an eigenvalue of the operator $H^{(0)}$ on this
modified graph $\Gamma'$. Now, if one repeats the arguments leading
to (\ref{eq: amplitude_relation_1-1}),(\ref{eq: amplitude_relation_2-1})
one concludes that (\ref{eq: amplitude_relation_1-1}),(\ref{eq: amplitude_relation_2-1})
hold for all $\vec{\kappa}\in\sreg$.

\subsection{Integrating over the secular manifold -- the Barra-Gaspard measure\label{subsec:Integrating-over-secular-manifold}}
\begin{defn}
\label{def:bg-measure}(\cite{BarGas_jsp00,BerWin_tams10,CdV_ahp15}).
The \emph{Barra-Gaspard measure} on $\sreg$ is the Radon probability
measure
\begin{equation}
\rmd\mu_{\vec{\ell}}\left(\vec{\kappa}\right)=\frac{\pi}{\left|\Gamma\right|}\cdot\frac{1}{\left(2\pi\right)^{E}}\left|\hat{n}\left(\vec{\kappa}\right)\cdot\vec{\ell}\right|\rmd s,\label{eq:-5}
\end{equation}
where $ds$ is the Lebesgue surface element and $\hat{n}$ is the
unit normal to the secular manifold. For a thorough introduction to
the Barra-Gaspard measure, see \cite{Alon_PhDThesis,AloBanBer_cmp18,AloBanBer_21arxiv,BarGas_jsp00,BerWin_tams10,CdV_ahp15}. 
\end{defn}

\begin{rem}
\label{rem:meas0} The singular set, $\Sigma\backslash\sreg$ is of
codimension at least one in $\Sigma$ \cite[thm 1.1]{CdV_ahp15},\cite[thm 3.6]{AloBanBer_cmp18},
and is thus of measure zero (since $\mu_{\vec{\ell}}$ is a Radon
measure). Thus, for the purpose of integration, we may as well extend
the Barra-Gaspard measure $\mu_{\vec{\ell}}$ to be defined over the
entire secular manifold $\Sigma$ (by setting it to be zero on the
set $\Sigma\backslash\sreg$). At any rate, the functions we will
consider and would like to integrate are only well defined on $\sreg$.
\end{rem}

The following ergodic theorem will be a main tool in proving our results:
\begin{thm}
\label{thm:ergodic-theorem}(\cite{BarGas_jsp00,BerWin_tams10,CdV_ahp15}).
Assume that the entries of the vector $\vec{\ell}$ are linearly independent
over $\mathbb{Q}$. Let $g$ be a Riemann integrable function on $\sreg$
(equivalently, $g$ is continuous almost everywhere and bounded).
Then:
\begin{equation}
\left\langle g\right\rangle _{n}:=\lim_{N\rightarrow\infty}\frac{1}{N}\sum_{n=1}^{N}g\left(k_{n}\vec{\ell}\right)=\int_{\sreg}g\left(\vec{\kappa}\right)\rmd\mu_{\vec{\ell}}\left(\vec{\kappa}\right).\label{eq:-7}
\end{equation}
\end{thm}

Thus, given a Riemann integrable function $g$ on the secular manifold,
the ergodic theorem allows us to compute the Ces\`{a}ro mean of the
sequence $\left(g\left(k_{n}\vec{\ell}\right)\right)_{n=1}^{\infty}$.
This will be the main idea behind the proof of Theorem \ref{thm:mean_RNG}.

In addition to Theorem \ref{thm:ergodic-theorem}, we will need the
following result, which gives further information about the manifold
$\sreg$ and the integration measure $\mu_{\vec{\ell}}:$
\begin{lem}
\label{lem:secular-normal}(\cite{AloBanBer_cmp18,CdV_ahp15}). Let
$\vec{\kappa}\in\sreg$ and denote $\boldsymbol{\kappa}:=\mathrm{diag}\left(\kappa_{1},\kappa_{1},...,\kappa_{E},\kappa_{E}\right)$.
Let $\vec{a}\in\ker\left(I-S\rme^{\rmi\boldsymbol{\kappa}}\right)$
be $\C^{2E}$ normalized, i.e., $\left\Vert \vec{a}\right\Vert _{2}=1$.
Let $\hat{n}\in\R^{E}$ be the unit normal to the secular manifold
at the point $\vec{\kappa}$.\\
Then, $\hat{n}$ is given by 
\begin{equation}
\forall e\in\E,\quad\quad\hat{n}_{e}=\left|a_{e}\right|^{2}+\left|a_{\hat{e}}\right|^{2}.\label{eq:-8}
\end{equation}
\end{lem}

\bigskip

\section{Proof of Theorem \ref{thm:Lipschitz_RNG} \label{sec:Lipschitz-proof}}

Theorem \ref{thm:Lipschitz_RNG} follows quite easily as a corollary
of Lemma \ref{lem:hadamard} and the following Lemma.
\begin{lem}
\label{lem:bdd} Let $f_{n}^{\left(t\right)}$ be an $L^{2}$ normalized
eigenfunction of $H^{(t)}$ with a simple eigenvalue. Then for every
$v\in\mathcal{V}_{R}$, the value of $\left|f_{n}^{\left(t\right)}\left(v\right)\right|^{2}$
is uniformly bounded in $n\in\mathbb{N}$ and $t\in\mathbb{R}$.
\end{lem}

\begin{proof}
Let $v\in\VR$. Assume that $v$ is an endpoint of an edge $j\in\E$
and that the parametrization sets $v$ to be at $x=0$ (This assumption
is for convenience, and will be used later as well). Denote by $\widetilde{f}_{n}^{\left(t\right)}$
the eigenfunction of $H^{(t)}$ whose coefficient vector $\vec{a}$
is $\mathbb{C}^{2E}$ normalized, $\left\Vert \vec{a}\right\Vert _{2}=1$.
Since the corresponding eigenvalue is simple, $\widetilde{f}_{n}^{\left(t\right)}$
and $f_{n}^{\left(t\right)}$ equal up to a multiplicative constant.
As $f_{n}^{\left(t\right)}$ is $L^{2}$ normalized, we can write
\begin{align}
\left|f_{n}^{\left(t\right)}\left(0\right)\right|^{2}=\frac{\left|\widetilde{f}_{n}^{\left(t\right)}\left(0\right)\right|^{2}}{\left\Vert \widetilde{f}_{n}^{\left(t\right)}\right\Vert _{L^{2}}^{2}}=\frac{\left|a_{j}+a_{\hat{j}}\rme^{\rmi k_{n}\ell_{j}}\right|^{2}}{\sum_{e=1}^{E}\left[\ell_{e}\left(\left|a_{e}\right|^{2}+\left|a_{\hat{e}}\right|^{2}\right)+\frac{2}{k_{n}}Re\left(a_{e}\overline{a_{\hat{e}}}\thinspace\rme^{-\rmi k_{n}\ell_{e}}\right)\right]}.\label{eq:-9}
\end{align}
We have used above the expression for $\widetilde{f}_{n}^{\left(t\right)}$
given in (\ref{eq:-6}) and performed a straightforward integration
to obtain the denominator in (\ref{eq:-9}). Note that the coefficients
$\left\{ a_{e},a_{\hat{e}}\right\} _{e\in\E}$ depend on $n$ and
$t$, but we suppressed this dependence above for ease of notation.
Since $a_{e}$ and $a_{\hat{e}}$ are all bounded in absolute value
by one and $k_{n}\rightarrow\infty$, the right hand side of (\ref{eq:-9})
is uniformly bounded in $n\in\mathbb{N}$ and $t\in\mathbb{R}$, which
completes the proof (the interested reader may find more details in
the proof of lemma 5.2 in \cite{Sofer2022}).
\end{proof}
\bigskip

The proof of Theorem \ref{thm:Lipschitz_RNG} now follows easily:
\begin{proof}[Proof of Theorem \ref{thm:Lipschitz_RNG}]
 By Lemma \ref{lem:hadamard}:
\begin{equation}
d_{n}\left(\sigma\right)=\sum_{v\in\mathcal{V}_{\mathcal{R}}}\int_{0}^{\sigma}\left|f_{n}^{\left(t\right)}\left(v\right)\right|^{2}\rmd t.\label{eq:-41}
\end{equation}
Lemma \ref{lem:bdd} then shows that $d_{n}'\left(\sigma\right)=\sum_{v\in\mathcal{V}_{\mathcal{R}}}\left|f_{n}^{\left(t\right)}\left(v\right)\right|^{2}$
is uniformly bounded in $n$ and $t$, which means that the all functions
in the sequence $\left(d_{n}\left(\sigma\right)\right)_{n=1}^{\infty}$
are Lipschitz continuous in $[0,\infty)$ with a uniform Lipschitz
constant. In particular, this implies that the sequence of functions
$\left(d_{n}\left(\sigma\right)\right)_{n=1}^{\infty}$ is uniformly
bounded on any compact subset $A\subset[0,\infty)$.
\end{proof}
\bigskip

\section{Proof of Theorem \ref{thm:Weyl-law} \label{sec:Proof-Weyl}}

All statements in Sections \ref{sec:Proof-Weyl} and \ref{sec:mean-proof}
are proven under the following assumption:\\

\begin{assumption}
\label{assumption: rationality} The edge lengths of the graph are
linearly independent over $\Q$.
\end{assumption}

This assumption allows to conveniently apply Theorem \ref{thm:ergodic-theorem}
in all proofs of the current and the next section. Nevertheless, all
the statements are valid even without this assumption. Indeed, in
Appendix \ref{sec:Appendix-A} we use a simple continuity argument
to show that the assumption can be omitted.

In the course of proving Theorem \ref{thm:Weyl-law}, we first express
the second moments of the eigenfunction scattering amplitudes in two
separate lemmas, and then use these expressions to prove the local
Weyl law.
\begin{lem}
\label{lem:Amean} Assume that the eigenfunction scattering amplitudes
are $\C^{2E}$ normalized, i.e., $\left\Vert \vec{a}\right\Vert _{2}=1$.
The following holds for all $j\in\E$:
\begin{equation}
\left\langle \frac{\left|a_{j}\right|^{2}}{\sum_{e=1}^{E}\ell_{e}\left(\left|a_{e}\right|^{2}+\left|a_{\hat{e}}\right|^{2}\right)}\right\rangle _{n}=\left\langle \frac{\left|a_{\hat{j}}\right|^{2}}{\sum_{e=1}^{E}\ell_{e}\left(\left|a_{e}\right|^{2}+\left|a_{\hat{e}}\right|^{2}\right)}\right\rangle _{n}=\frac{1}{2\left|\Gamma\right|},\label{eq:-42-1-1}
\end{equation}
where the $n$-dependence of the amplitudes, $(a_{j})_{n}$, is as
given in (\ref{eq:-45}).
\end{lem}

\begin{proof}
In order to express the mean value in the left hand side of (\ref{eq:-42-1-1}),
we apply the ergodic theorem (Theorem \ref{thm:ergodic-theorem})
for the function
\begin{equation}
g_{j}\left(\vec{\kappa}\right):=\frac{\left|a_{j}(\vec{\kappa})\right|^{2}}{\sum_{e=1}^{E}\ell_{e}\left(\left|a_{e}(\vec{\kappa})\right|^{2}+\left|a_{\hat{e}}(\vec{\kappa})\right|^{2}\right)},\label{eq:-19}
\end{equation}
where $\vec{a}(\vec{\kappa})\in\ker\left(I-S\rme^{\rmi\boldsymbol{\kappa}}\right)$
and $\boldsymbol{\kappa}:=\mathrm{diag}\left(\kappa_{1},\kappa_{1},...,\kappa_{E},\kappa_{E}\right)$.
Furthermore, $\vec{a}(\vec{\kappa})$ is a $\mathbb{C}^{2E}$ normalized
vector which is chosen as described in Section \ref{subsec:The-secular-manifold}
and satisfies (\ref{eq: amplitude_relation_1-1}),(\ref{eq: amplitude_relation_2-1}).
Applying (\ref{eq:-7}) from the ergodic theorem gives
\begin{equation}
\left\langle g_{j}\right\rangle _{n}=\int_{\sreg}g_{j}\left(\vec{\kappa}\right)\rmd\mu_{\vec{\ell}}\left(\vec{\kappa}\right),\label{eq:ergodic-g}
\end{equation}
where $\left(g_{j}\right)_{n}:=g_{j}\left(k_{n}\vec{\ell}\right)$
and $\left(a_{j}\right)_{n}=a_{j}\left(k_{n}\vec{\ell}\right)$. We
may indeed apply the ergodic theorem, since $\left|a_{j}(\vec{\kappa})\right|^{2}$
are real analytic functions on $\sreg$ (see proof of lemma 4.25 in
\cite{Alon_PhDThesis}) and the denominator of (\ref{eq:-45}) is
bounded from below by a positive number, and so $g_{j}$ is Riemann
integrable. We thus get that the mean values in (\ref{eq:-42-1-1})
are well defined. By (\ref{eq: amplitude_relation_1-1}), we further
see that $g_{j}=g_{\hat{j}}$ and conclude

\begin{equation}
\left\langle \frac{\left|a_{j}\right|^{2}}{\sum_{e=1}^{E}\ell_{e}\left(\left|a_{e}\right|^{2}+\left|a_{\hat{e}}\right|^{2}\right)}\right\rangle _{n}=\left\langle \frac{\left|a_{\hat{j}}\right|^{2}}{\sum_{e=1}^{E}\ell_{e}\left(\left|a_{e}\right|^{2}+\left|a_{\hat{e}}\right|^{2}\right)}\right\rangle _{n}.\label{eq:-34-2}
\end{equation}
 This proves the first equality in the lemma, and all that remains
is to show that both terms above are equal to $\frac{1}{2\left|\Gamma\right|}$.
We do so by proving that 
\begin{equation}
\left\langle \frac{\left|a_{j}\right|^{2}+\left|a_{\hat{j}}\right|^{2}}{\sum_{e=1}^{E}\ell_{e}\left(\left|a_{e}\right|^{2}+\left|a_{\hat{e}}\right|^{2}\right)}\right\rangle _{n}=\frac{1}{\left|\Gamma\right|}.\label{eq:-35-2}
\end{equation}
We calculate,

\begin{align}
\left\langle \frac{\left|a_{j}\right|^{2}+\left|a_{\hat{j}}\right|^{2}}{\sum_{e=1}^{E}\ell_{e}\left(\left|a_{e}\right|^{2}+\left|a_{\hat{e}}\right|^{2}\right)}\right\rangle _{n}=\left\langle \frac{\hat{n}_{j}}{\sum_{e=1}^{E}\ell_{e}\hat{n}_{e}}\right\rangle _{n} & =\label{eq:-36-2}\\
=\int_{\sreg}\frac{\hat{n}_{j}}{\sum_{e=1}^{E}\ell_{e}\hat{n}_{e}}\rmd\mu_{\vec{\ell}}=\frac{\pi}{\left|\Gamma\right|}\cdot\frac{1}{\left(2\pi\right)^{E}}\int_{\Sigma}\hat{n}_{j}\rmd s,
\end{align}
where the first equality follows from the expression of the unit normal
to $\Sigma^{\mathrm{reg}}$, (\ref{eq:-8}), the second equality is
(\ref{eq:-7}) in the ergodic theorem, and the third is by the definition
of the Barra-Gaspard measure (\ref{eq:-5}) and Remark \ref{rem:meas0}.
Hence, (\ref{eq:-35-2}) is equivalent to showing that $\int_{\Sigma}\hat{n}_{j}\rmd s=2\left(2\pi\right)^{E-1}$,
which is what we prove next.

The given integral is exactly the surface integral of the vector field
$\frac{\partial}{\partial x_{j}}$ over the secular manifold, $\Sigma$.
We consider the projection $\pi_{j}:\Sigma\rightarrow\left(\mathbb{R}/2\pi\mathbb{Z}\right)^{E-1}$,
which is defined by omitting the $j^{\textrm{th }}$ coordinate of
$\vec{\kappa}$. The proof of proposition $3.1$ in \cite{CdV_ahp15}
shows that this projection is a two to one map. Namely, $\left|\pi_{j}^{-1}(\vec{x})\right|=2$
for all $\vec{x}\in\left(\mathbb{R}/2\pi\mathbb{Z}\right)^{E-1}$.
Since the surface integral of $\frac{\partial}{\partial x_{j}}$ is
invariant under this projection, we get that the given integral is
equal to twice the flux of the vector field $\frac{\partial}{\partial x_{j}}$
through the $j^{\textrm{th}}$ face of the torus:
\begin{equation}
\int_{\Sigma}\hat{n}_{j}\rmd s=2\int_{\left(\mathbb{R}/2\pi\mathbb{Z}\right)^{E-1}}1\thinspace\rmd s=2\thinspace\mathrm{vol}\left(\left(\mathbb{R}/2\pi\mathbb{Z}\right)^{E-1}\right)=2\left(2\pi\right)^{E-1},\label{eq:-37-2}
\end{equation}
as required.
\end{proof}
\begin{lem}
\label{lem:Uncorrelation} Assume that the eigenfunction scattering
amplitudes are $\C^{2E}$ normalized, i.e., $\left\Vert \vec{a}\right\Vert _{2}=1$.
The following holds for all $i,j\in\E$, $i\neq j$:
\begin{equation}
\left\langle \frac{a_{i}\overline{a_{j}}}{\sum_{e=1}^{E}\ell_{e}\left(\left|a_{e}\right|^{2}+\left|a_{\hat{e}}\right|^{2}\right)}\right\rangle _{n}=0,\label{eq:-43-2}
\end{equation}
where the $n$-dependence of the amplitudes, $(a_{j})_{n}$, is as
given in (\ref{eq:-45}).
\end{lem}

\begin{proof}
By the ergodic theorem (Theorem \ref{thm:ergodic-theorem}), combined
with Remark \ref{rem:meas0}, we have that
\begin{equation}
\left\langle g\right\rangle _{n}=\int_{\sreg}g\rmd\mu_{\vec{\ell}}=\int_{\Sigma}g\rmd\mu_{\vec{\ell}}.\label{eq:ergodicg2}
\end{equation}
We now refer to \cite[thm. 4.10]{AloBanBer_21arxiv}, which gives
an alternative method for integrating functions over $\Sigma$:
\begin{equation}
\int_{\Sigma}g\rmd\mu_{\vec{\ell}}=\int_{\mathbb{T}^{E}}\sum_{m=1}^{2E}g\left(\vec{\kappa}-\theta_{m}\cdot\vec{1}\right)\frac{\left(\vec{a}^{(m)}\right)^{*}\L\vec{a}^{(m)}}{\tr(L)}\frac{\rmd\vec{\kappa}}{\left(2\pi\right)^{E}},\label{eq:-11-2}
\end{equation}
where $(\rme^{\rmi\theta_{m}(\vec{\kappa})})_{m=1}^{2E}$ are the
eigenvalues of the unitary matrix $S\rme^{\rmi\boldsymbol{\kappa}}$
, $(\vec{a}^{(m)}(\vec{\kappa}))_{m=1}^{2E}$ are its $\mathbb{C}^{2E}$
normalized eigenvectors, and $\vec{1}=(1,\ldots,1)$. Equation (\ref{eq:-11-2})
is useful, as it allows to replace integration over the secular manifold
with integration over the whole torus, if the spectral decomposition
of $S\rme^{\rmi\boldsymbol{\kappa}}$ is known.

Next, we apply (\ref{eq:-11-2}) for the function 
\begin{equation}
g\left(\vec{\kappa}\right):=\frac{a_{i}\overline{a_{j}}}{\sum_{e=1}^{E}\ell_{e}\left(\left|a_{e}\right|^{2}+\left|a_{\hat{e}}\right|^{2}\right)},\label{eq:gfunction4.3}
\end{equation}
noting that the vector $\vec{a}(\vec{\kappa})$ which appears in $g$
(and in the statement of the lemma) is one of the eigenvectors $\vec{a}^{(n)}$
of the unitary matrix $S\rme^{\rmi\boldsymbol{\kappa}}$ (in particular
it is an eigenvector which corresponds to the eigenvalue $1$). As
a matter of fact, the set of eigenvectors $\left\{ \vec{a}^{(m)}\right\} _{m=1}^{2E}$
of $S\rme^{\rmi\boldsymbol{\kappa}}$ is exactly $\left\{ \vec{a}(\vec{\kappa}-\theta_{i}\cdot\vec{1})\right\} _{i=1}^{2E}$
(up to possible reordering, which we do not care about here). We further
observe that $\tr(L)=2\left|\Gamma\right|$ and $\left(\vec{a}^{(m)}\right)^{*}\L\vec{a}^{(m)}=\sum_{e=1}^{E}\ell_{e}\left(\left|a_{e}^{(m)}\right|^{2}+\left|a_{\hat{e}}^{(m)}\right|^{2}\right)$.
Using all of the above, we get 
\begin{align}
\left\langle \frac{a_{i}\overline{a_{j}}}{\sum_{e=1}^{E}\ell_{e}\left(\left|a_{e}\right|^{2}+\left|a_{\hat{e}}\right|^{2}\right)}\right\rangle _{n} & =\frac{1}{2\left|\Gamma\right|\left(2\pi\right)^{E}}\int_{\mathbb{T}^{E}}\sum_{n=1}^{2E}a_{i}\left(\vec{\kappa}-\theta_{n}\right)\thinspace\overline{a_{j}}\left(\vec{\kappa}-\theta_{n}\right)\thinspace\rmd\vec{\kappa}.\label{eq:-38-2-1}\\
 & =\frac{1}{2\left|\Gamma\right|\left(2\pi\right)^{E}}\int_{\mathbb{T}^{E}}\sum_{n=1}^{2E}a_{i}^{(n)}\thinspace\overline{a_{j}^{(n)}}\thinspace\rmd\vec{\kappa}=0,
\end{align}
where the last equality follows since the eigenvectors $\vec{a}^{(n)}$
of the matrix $S\rme^{\rmi\boldsymbol{\kappa}}$ are orthogonal (so
the integrand itself in fact vanishes identically).
\end{proof}
\begin{proof}[Proof of Theorem \ref{thm:Weyl-law} ]

We note that Theorem \ref{thm:Weyl-law} is stated for $L^{2}$ normalized
eigenfunctions. We have already calculated (see (\ref{eq:-9})) that
the $L^{2}$ norm of an eigenfunction is 
\begin{equation}
\left\Vert f_{n}\right\Vert ^{2}=\sum_{e=1}^{E}\left[\ell_{e}\left(\left|(a_{e})_{n}\right|^{2}+\left|(a_{\hat{e}})_{n}\right|^{2}\right)+O\left(\frac{1}{k_{n}}\right)\right].\label{eq:eigenfunctionnorm}
\end{equation}
Since terms of the form $O\left(\frac{1}{k_{n}}\right)$ do not affect
the Ces\`{a}ro mean, we get that Lemmas \ref{lem:Amean} and \ref{lem:Uncorrelation}
prove exactly the expressions (\ref{eq:-16}) and (\ref{eq:-17})
in the Theorem. We proceed to prove the local Weyl law, (\ref{eq:-15-1}).

Take $j\in\E$ to be an edge connected to the vertex $v$. We start
by employing the expression for $\left|f_{n}\left(v\right)\right|^{2}$
which was already computed in the proof of Lemma \ref{lem:bdd} (see
(\ref{eq:-9})).
\begin{align}
\left\langle \left|f(v)\right|^{2}\right\rangle _{n} & =\left\langle \frac{\left|a_{j}+a_{\hat{j}}\text{e}^{\rmi k_{n}\ell_{j}}\right|^{2}}{\sum_{e=1}^{E}\left[\ell_{e}\left(\left|a_{e}\right|^{2}+\left|a_{\hat{e}}\right|^{2}\right)+O\left(\frac{1}{k_{n}}\right)\right]}\right\rangle _{n}\label{eq:-61-1-1}\\
 & =\left\langle \frac{\left|a_{j}+a_{\hat{j}}\text{e}^{\rmi k_{n}\ell_{j}}\right|^{2}}{\sum_{e=1}^{E}\ell_{e}\left(\left|a_{e}\right|^{2}+\left|a_{\hat{e}}\right|^{2}\right)}\right\rangle _{n}\label{eq: sqaure of eigenfunction on secular manifold}\\
 & =\left\langle \frac{\left|a_{j}\right|^{2}+\left|a_{\hat{j}}\right|^{2}}{\sum_{e=1}^{E}\ell_{e}\left(\left|a_{e}\right|^{2}+\left|a_{\hat{e}}\right|^{2}\right)}\right\rangle _{n}+2\left\langle \frac{Re\left(\text{e}^{\rmi k_{n}\ell_{j}}a_{\hat{j}}\overline{a_{j}}\right)}{\sum_{e=1}^{E}\ell_{e}\left(\left|a_{e}\right|^{2}+\left|a_{\hat{e}}\right|^{2}\right)}\right\rangle _{n},\label{eq:-14-2}
\end{align}
where going to the second line we omitted the term $O\left(\frac{1}{k_{n}}\right)$
in the Ces\`{a}ro mean for the same reason as above. By Lemma \ref{lem:Amean},
the first term in (\ref{eq:-14-2}) equals $\frac{1}{\left|\Gamma\right|}$,
and we proceed to evaluate the second term. From $(\vec{a})_{n}\in\ker\left(I-S\rme^{\rmi k_{n}\L}\right)$
and the expression for $S$ in (\ref{eq:-32-2}), we get that
\begin{equation}
(a_{j})_{n}=\text{e}^{\rmi k_{n}\ell_{j}}\sum_{i=1}^{\deg\left(v\right)}\left(\frac{2}{\deg\left(v\right)}-\delta_{ji}\right)(a_{\hat{i}})_{n}.\label{eq:-39-2}
\end{equation}
Thus:
\begin{align}
\left\langle \frac{Re\left(\text{e}^{\rmi k_{n}\ell_{j}}a_{\hat{j}}\overline{a_{j}}\right)}{\sum_{e=1}^{E}\ell_{e}\left(\left|a_{e}\right|^{2}+\left|a_{\hat{e}}\right|^{2}\right)}\right\rangle _{n}= & \left(\frac{2}{\deg\left(v\right)}-1\right)\left\langle \frac{\left|a_{\hat{j}}\right|^{2}}{\sum_{e=1}^{E}\ell_{e}\left(\left|a_{e}\right|^{2}+\left|a_{\hat{e}}\right|^{2}\right)}\right\rangle _{n}\label{eq:-62-1-1}\\
 & +\frac{2}{\deg\left(v\right)}\sum_{i\neq j\in\mathcal{E}_{v}}Re\left\langle \frac{a_{\hat{j}}\overline{a_{\hat{i}}}}{\sum_{e=1}^{E}\ell_{e}\left(\left|a_{e}\right|^{2}+\left|a_{\hat{e}}\right|^{2}\right)}\right\rangle _{n}.\label{eq:-40-2-1}
\end{align}
By Lemmas \ref{lem:Amean} and \ref{lem:Uncorrelation}, the first
term is equal to $\frac{1}{2\left|\Gamma\right|}\left(\frac{2}{\deg\left(v\right)}-1\right)$,
while the second term is equal to zero. Plugging this into (\ref{eq:-14-2}),
we obtain:
\begin{equation}
\left\langle \left|f\left(v\right)\right|^{2}\right\rangle _{n}=\frac{1}{\left|\Gamma\right|}+\frac{1}{\left|\Gamma\right|}\left(\frac{2}{\deg\left(v\right)}-1\right)=\frac{2}{\deg\left(v\right)\left|\Gamma\right|}.\label{eq:-41-1-1}
\end{equation}
\end{proof}
\bigskip

\section{Proof of Theorem \ref{thm:mean_RNG}\label{sec:mean-proof}}

To prove Theorem \ref{thm:mean_RNG}, we use the following two lemmas.
In both lemmas, we denote by $f_{n}^{\left(t\right)}$ the $n^{\textrm{th}}$
$L^{2}$ normalized eigenfunction of $H^{(t)}$ (see also Lemma \ref{lem:hadamard}).
\begin{lem}
\label{lem:eigenfunc-conv} For $\sigma>0$ fixed, $\left|f_{n}^{\left(t\right)}\left(v\right)-f_{n}^{\left(0\right)}\left(v\right)\right|\underset{_{n\rightarrow\infty}}{\longrightarrow}0$
uniformly in $t\in\left[0,\sigma\right]$.
\end{lem}

\begin{lem}
\label{lem:eigmean-conv} Fix $\sigma>0$. Then
\begin{equation}
\lim_{N\rightarrow\infty}\frac{1}{N}\sum_{n=1}^{N}\left|f_{n}^{\left(t\right)}\left(v\right)\right|^{2}=\lim_{N\rightarrow\infty}\frac{1}{N}\sum_{n=1}^{N}\left|f_{n}^{\left(0\right)}\left(v\right)\right|^{2},\label{eq:5.1convergence}
\end{equation}
where the convergence is uniform in $t\in\left[0,\sigma\right]$.
\end{lem}

Before proving the lemmas, we use them to prove Theorem \ref{thm:mean_RNG}.
\begin{proof}[Proof of Theorem \ref{thm:mean_RNG}]
 Within this proof, we apply Lemma \ref{lem:hadamard}, and do so
for a single Robin vertex $v$. The additivity of (\ref{eq:-1}) implies
that the proof holds for any set $\VR$ of Robin vertices (see Remark
\ref{rem:additivity}).

\begin{align}
\left\langle d\right\rangle _{n}\left(\sigma\right) & =\lim_{N\rightarrow\infty}\frac{1}{N}\sum_{n=1}^{N}d_{n}\left(\sigma\right)\label{eq:-50}\\
 & =_{(\ref{eq:-1})}\lim_{N\rightarrow\infty}\frac{1}{N}\sum_{n=1}^{N}\int_{0}^{\sigma}\left|f_{n}^{\left(t\right)}\left(v\right)\right|^{2}\rmd t\label{eq:RNG expressed as t-integral}\\
 & =\lim_{N\rightarrow\infty}\int_{0}^{\sigma}\frac{1}{N}\sum_{n=1}^{N}\left|f_{n}^{\left(t\right)}\left(v\right)\right|^{2}\rmd t\\
 & =\int_{0}^{\sigma}\lim_{N\rightarrow\infty}\frac{1}{N}\sum_{n=1}^{N}\left|f_{n}^{\left(0\right)}\left(v\right)\right|^{2}\rmd t\\
 & =\left\langle \left|f^{\left(0\right)}\left(v\right)\right|^{2}\right\rangle _{n}\sigma=\frac{2}{\deg\left(v\right)\left|\Gamma\right|}\sigma.\label{eq: RNG as square value of eigenfunction}
\end{align}
where we have used Lemma \ref{lem:eigmean-conv} when moving to the
fourth line, and the local Weyl law (\ref{eq:-15}) in the last equality.
\end{proof}
\begin{proof}[Proof of Lemma \ref{lem:eigenfunc-conv}]
 By Theorem \ref{thm:secular-equation}, the eigenvalues of $H^{\left(t\right)}$
are $\left(k_{n}^{(t)}\right)^{2}$, where $k_{n}^{(t)}$ are the
solutions to the secular equation. Recall that the roots $k_{n}^{(t)}$,
determine the coefficients of the eigenfunction via the (analytic)
matrix equation (\ref{eq:-2}) and (\ref{eq:-6}). It is thus enough
to show that as $n\rightarrow\infty$, $\left|k_{n}^{(t)}-k_{n}^{(0)}\right|\rightarrow0$
uniformly in $t\in\left[0,\sigma\right]$. 

By Theorem \ref{thm:secular-equation}, $k_{n}^{(t)}$ are the $k$
values for which $\det\left(I-S^{(t)}\rme^{\rmi k\L}\right)=0$, where
\begin{equation}
S_{j'j}^{\left(t\right)}=\begin{cases}
\frac{2}{\deg\left(v\right)+\frac{\rmi t}{k}}-1 & j'=\hat{j}\\
\frac{2}{\deg\left(v\right)+\frac{\rmi t}{k}} & j\rightarrow j'\text{ at \ensuremath{v} and }j'\neq\hat{j}\\
0 & \text{Otherwise},
\end{cases}\label{eq:-3-1-1-1-1-1}
\end{equation}
as in (\ref{eq:-32-2}). Denoting $U^{\left(t\right)}\left(k\right):=S^{\left(t\right)}\text{e}^{ik\L}$,
we get 
\begin{equation}
\left\Vert U^{\left(t\right)}\left(k\right)-U^{\left(0\right)}\left(k\right)\right\Vert _{\infty}=\left\Vert \text{e}^{\rmi k\L}\left(S^{\left(t\right)}-S^{\left(0\right)}\right)\right\Vert _{\infty}\leq\frac{2t}{\deg\left(v\right)\left|\deg\left(v\right)+\frac{\rmi t}{k}\right|k},\label{eq:-52}
\end{equation}
and this expression approaches zero uniformly in $t\in\left[0,\sigma\right]$
as $k\rightarrow\infty$ . Since the supremum norm of the difference
tends to zero, so does the operator norm of the difference. This means
that as $k\rightarrow\infty$, the eigenvalues of $U^{\left(t\right)}\left(k\right)$
converge to those of $U^{\left(0\right)}\left(k\right)$ uniformly
in $t\in\left[0,\sigma\right]$. Denote the eigenvalues of the unitary
matrix $U^{\left(t\right)}\left(k\right)$ by $\left(\text{e}^{i\theta_{m}^{(t)}\left(k\right)}\right)_{m=1}^{2E}$,
so that the eigenphases $\left(\theta_{m}^{(t)}\left(k\right)\right)_{m=1}^{2E}$
are the lifts of these eigenvalues from $S^{1}$ to the universal
cover $\mathbb{R}$. By Theorem \ref{thm:secular-equation}, $k^{2}>0$
is an eigenvalue of $H^{(t)}$ if and only if $\theta_{m}^{(t)}\left(k\right)\in2\pi\mathbb{Z}$
for some $m$. Denote by $\left(k_{n}^{(t)}\right)_{n=1}^{\infty}$
the $k$ values for which this happens (these are exactly the zeros
of the secular function $\det\left(I-S^{(t)}\rme^{\rmi k\L}\right)$).

We know that $\left(\theta_{m}^{(t)}\left(k\right)\right)_{m=1}^{2E}$
increase monotonically with $k$ at a rate which is bounded from below
by some $c>0$, and that $k_{n}^{(0)}<k_{n}^{(t)}$ for all $n$,
\cite[lem. 4.5]{BolEnd_ahp09}. Then by applying the mean value theorem,
we get
\begin{align}
 & \theta_{m}^{(0)}\left(k_{n}^{(0)}\right)=\theta_{m}^{(t)}\left(k_{n}^{(t)}\right)\geq\theta_{m}^{(t)}\left(k_{n}^{(0)}\right)+c\left(k_{n}^{(t)}-k_{n}^{(0)}\right)\label{eq:-53}\\
\Rightarrow & ~~k_{n}^{(t)}-k_{n}^{(0)}\leq\frac{1}{c}\left(\theta_{m}^{(0)}\left(k_{n}^{(0)}\right)-\theta_{m}^{(t)}\left(k_{n}^{(0)}\right)\right).\label{eq:-54}
\end{align}

As $n\rightarrow\infty$ (which is equivalent to $k\rightarrow\infty$),
we know that $\theta_{m}^{(0)}\left(k_{n}^{(0)}\right)-\theta_{m}^{(t)}\left(k_{n}^{(0)}\right)\longrightarrow0$
uniformly in $t\in\left[0,\sigma\right]$ (as argued (\ref{eq:-52})).
Therefore, we conclude that as $n\rightarrow\infty$, $\left|k_{n}^{(t)}-k_{n}^{(0)}\right|\rightarrow0$
uniformly in $t\in\left[0,\sigma\right]$. This completes the proof.
\end{proof}
\begin{proof}
[Proof of Lemma \ref{lem:eigmean-conv}]

We first note that for $t=0$, $\left\langle \left|f^{\left(0\right)}\left(v\right)\right|^{2}\right\rangle _{n}$
exists by the local Weyl law in Theorem \ref{thm:Weyl-law}. We denote
this mean value by $C$ for brevity. For $t\neq0$, we claim that
$\left\langle \left|f^{\left(t\right)}\left(v\right)\right|^{2}\right\rangle _{n}$
exists as well, that it actually equals the same constant $C$, and
that the convergence is uniform with $t\in\left[0,\sigma\right]$.
Writing
\begin{equation}
\frac{1}{N}\sum_{n=1}^{N}\left|f_{n}^{\left(t\right)}\right|^{2}=\frac{1}{N}\sum_{n=1}^{N}\left|f_{n}^{\left(0\right)}\right|^{2}+\frac{1}{N}\sum_{n=1}^{N}\left(\left|f_{n}^{\left(t\right)}\right|^{2}-\left|f_{n}^{\left(0\right)}\right|^{2}\right),\label{eq:-60}
\end{equation}
we have that as $N\rightarrow\infty$, the first term converges to
$C$, and we claim that the second term converges to zero uniformly
with $t\in\left[0,\sigma\right]$.

Recall by Lemma \ref{lem:bdd} that the values of $\left|f_{n}^{\left(t\right)}\left(v\right)\right|^{2}$
are all uniformly bounded in $t\in\mathbb{R}$ by some $M>0$. Hence,
\begin{align}
 & \left|\left|f_{n}^{\left(t\right)}\left(v\right)\right|^{2}-\left|f_{n}^{\left(0\right)}\left(v\right)\right|^{2}\right|=\left|f_{n}^{\left(t\right)}\left(v\right)-f_{n}^{\left(0\right)}\left(v\right)\right|\cdot\left|f_{n}^{\left(t\right)}\left(v\right)+f_{n}^{\left(0\right)}\left(v\right)\right|\label{eq:-61-1}\\
 & \leq2M\left|f_{n}^{\left(t\right)}\left(v\right)-f_{n}^{\left(0\right)}\left(v\right)\right|\rightarrow0.\label{eq:-62-1}
\end{align}
This convergence is uniform with $t\in\left[0,\sigma\right]$, since
$\left|f_{n}^{\left(t\right)}-f_{n}^{\left(0\right)}\right|\underset{n\rightarrow\infty}{\rightarrow}0$
uniformly by Lemma \ref{lem:eigenfunc-conv}. Now, we get that in
the limit $N\rightarrow\infty$, the second term of (\ref{eq:-60})
is a Ces\`{a}ro mean of a bounded series which converges uniformly
to zero. Using this, it follows straightforwardly that the Ces\`{a}ro
mean itself converges uniformly to zero, as required.
\end{proof}

\section{Proof and discussion of Theorems \ref{thm:probability} and \ref{thm:converging_subsequence}\label{sec:conv-sub-proof}}
\begin{proof}[Proof of Theorem \ref{thm:probability}]
 As in the preceding proofs, we may consider the case of a single
Robin vertex parameterized by $v=0$. Using arguments similar to the
ones used to derive (\ref{eq: RNG as square value of eigenfunction})
within the proof of Theorem \ref{thm:mean_RNG}, one can show that
\begin{equation}
\lim_{N\rightarrow\infty}\frac{\#\left\{ n\leq N:d_{n}\left(\sigma\right)\leq x\right\} }{N}=\lim_{N\rightarrow\infty}\frac{\#\left\{ n\leq N:\sigma\left|f_{n}^{\left(0\right)}\left(0\right)\right|^{2}\leq x\right\} }{N}.\label{eq:123}
\end{equation}
Motivated by (\ref{eq: sqaure of eigenfunction on secular manifold})
in the proof of Theorem \ref{thm:Weyl-law}, we consider the following
auxiliary function on $\sreg$:
\begin{equation}
g\left(\vec{\kappa}\right):=\frac{\left|a_{j}+a_{\hat{j}}\text{e}^{\rmi\kappa_{j}}\right|^{2}}{\sum_{e=1}^{E}\ell_{e}\left(\left|a_{e}\right|^{2}+\left|a_{\hat{e}}\right|^{2}\right)}.\label{eq:-18}
\end{equation}
Denoting $\vec{\kappa}_{n}:=k_{n}\vec{\ell}$, we have that
\begin{equation}
\lim_{N\rightarrow\infty}\frac{\#\left\{ n\leq N:\sigma\left|f_{n}^{\left(0\right)}\left(0\right)\right|^{2}\leq x\right\} }{N}=\lim_{N\rightarrow\infty}\frac{\#\left\{ n\leq N:\sigma g\left(\vec{\kappa}_{n}\right)\leq x\right\} }{N}.\label{eq:-24}
\end{equation}

For a fixed $x\in\mathbb{R}$, define the following characteristic
function on $\sreg$:
\begin{equation}
\eta_{x}\left(\vec{\kappa}\right)=\begin{cases}
1 & g\left(\vec{\kappa}\right)\leq\frac{x}{\sigma}\\
0 & \text{Otherwise}
\end{cases}\label{eq:-21}
\end{equation}
We claim that $\eta_{x}\left(\vec{\kappa}\right)$ is Riemann integrable
for every $x\in\mathbb{R}$. As in the proof of Lemma \ref{lem:Amean},
we employ the proof of lemma 4.25 in \cite{Alon_PhDThesis} to get
that all functions of the form $\left|a_{j}\left(\vec{\kappa}\right)\right|^{2}$
and $,a_{j}\left(\vec{\kappa}\right)\overline{a_{\hat{j}}}\left(\vec{\kappa}\right)$
are real analytic functions on $\sreg$. Since the denominator in
(\ref{eq:-18}) is bounded from below by a positive value, we conclude
that $g$ is real analytic as well. We get by \cite[lem. 7.1]{AloBan_21ahp}
that for each connected component $\mathcal{M}$ of $\sreg$, either
$g$ is constant or its level sets are of measure zero. In particular,
either the sublevel set $\left\{ g\left(\vec{\kappa}\right)\leq\frac{x}{\sigma}\right\} \cap\mathcal{M}$
is $\mathcal{M}$ or it is a submanifold of $\mathcal{M}$. In both
cases it has a boundary of measure zero. This implies that $\eta_{x}\left(\vec{\kappa}\right)$
is indeed Riemann integrable for every $x$.

The Riemann integrability of $\eta_{x}\left(\vec{\kappa}\right)$
allows to apply the ergodic theorem if we assume that the graph edge
lengths are linearly independent over $\mathbb{Q}$. We proceed by
first making this assumption (which is lifted later), and apply (\ref{eq:123}),(\ref{eq:-24})
and Theorem \ref{thm:ergodic-theorem} to obtain
\begin{align}
F_{\sigma}\left(x\right) & =\lim_{N\rightarrow\infty}\frac{\#\left\{ n\leq N:d_{n}\left(\sigma\right)\leq x\right\} }{N}=\lim_{N\rightarrow\infty}\frac{\#\left\{ n\leq N:\sigma g\left(\vec{\kappa}_{n}\right)\leq x\right\} }{N}\label{eq:-22}\\
 & =\lim_{N\rightarrow\infty}\frac{1}{N}\sum_{n=1}^{N}\eta_{x}\left(\vec{\kappa}_{n}\right)=\int_{\sreg}\eta_{x}\left(\vec{\kappa}\right)d\mu_{\vec{\ell}}=\mu_{\vec{\ell}}\left(\vec{\kappa}:g\left(\vec{\kappa}\right)\leq\frac{x}{\sigma}\right).\label{eq:-23}
\end{align}
We thus see that $F_{\sigma}\left(x\right)$ in fact returns the Barra-Gaspard
measure of the sublevel sets of $g$. The fact that $F_{\sigma}$
is a cumulative distribution function follows from $F_{\sigma}\left(x\right)=\mu_{\vec{\ell}}\left(\vec{\kappa}:g\left(\vec{\kappa}\right)\leq\frac{x}{\sigma}\right)$.
Moreover, Theorem \ref{thm:explicit_bounds} shows that the associated
probability measure $\mu_{\sigma}$ must be supported on $\left[0,\frac{4\sigma}{\ell_{\min}}\right]$.

Now, we lift the assumption that the graph edge lengths are linearly
independent over $\Q$. If the edge lengths are linearly dependent,
then the associated torus flow (as defined Section \ref{subsec:The-secular-manifold})
is no longer dense. Denoting the dimension of $\text{span}_{\mathbb{Q}}\left\{ \ell_{e}\right\} _{e\in\mathcal{E}}$
by $D$, we see that the torus flow is in fact only dense in a $D$-dimensional
subtorus of the $E$-dimensional torus. By applying an appropriate
change of coordinates, we may assume that the torus flow is dense
when projected onto the first $D$ coordinates of the torus, and constant
when projected onto the last $E-D$ coordinates (see Figure \ref{fig:tflow}).
\begin{figure}
\includegraphics[scale=0.3]{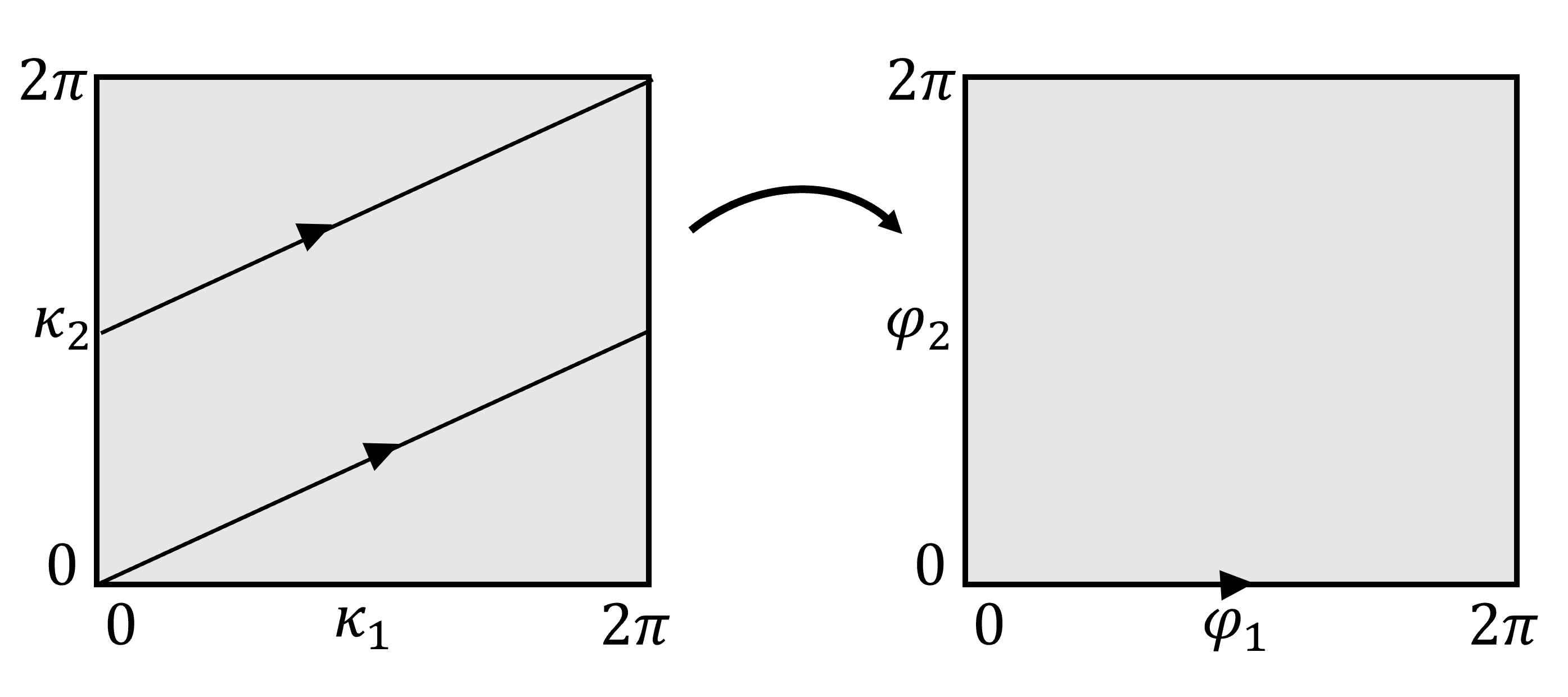}

\caption{\label{fig:tflow} The torus flow $\phi\left(k\right)=\left(2k,k\right)$
on $\mathbb{T}^{2}$ as an example of rationally dependent entries
(here $D=1$). After the change of coordinates $\left(\varphi_{1},\varphi_{2}\right)=\frac{1}{3}\left(\kappa_{1}+\kappa_{2},\kappa_{1}-2\kappa_{2}\right)$,
we get the flow $\tilde{\phi}\left(k\right)=\left(k,0\right)$, which
is dense in the first component and constant in the second component. }
\end{figure}
 Hence, we may consider the flow within this $D$-dimensional subtorus.
Repeating the same arguments as before, now considering the restriction
of the function $g$ to the intersection of $\sreg$ and the subtorus,
leads to the same conclusion -- that $F_{\sigma}$ is a cumulative
distribution function.

Lastly, we consider the case where $\ell_{e}/\ell_{e'}\in\mathbb{Q}$,
for all $e,e'\in\mathcal{E}$. In this case, $D=1$ and the subtorus
mentioned above is one-dimensional (see also Figure \ref{fig:tflow}).
Hence, the intersection of $\sreg$ and the subtorus consists of only
finitely many points. In particular the sequence $\left(g\left(\vec{\kappa}_{n}\right)\right)_{n=1}^{\infty}$
contains finitely many values (it is in fact also periodic). This
implies that $F_{\sigma}\left(x\right)$ only attains finitely many
values, and so the associated measure $\mu_{\sigma}$ is a convex
combination of Dirac masses.
\end{proof}
\begin{rem}
\label{rem:Absolutely-continuous-conjecture} We conjecture that if
not all ratios of edge lengths are rational, then the measure $\mu_{\sigma}$
is absolutely continuous with respect to the Lebesgue measure. We
briefly explain the intuition behind this conjecture. Assuming that
the function $g$ from the proof above only attains regular values,
then the co-area formula gives
\begin{equation}
F_{\sigma}\left(x\right)=\mu_{\vec{\ell}}\left(\vec{\kappa}:g\left(\vec{\kappa}\right)\leq\frac{x}{\sigma}\right)=\int_{-\infty}^{x/\sigma}dt\int_{\left\{ g=t\right\} }\frac{1}{\left|\nabla g\right|}d\mu_{\vec{\ell}}^{t}\left(\vec{\kappa}\right),\label{eq:co-area}
\end{equation}
where $d\mu_{\vec{\ell}}^{t}$ denotes the surface area element of
the Barra-Gaspard measure $\mu_{\vec{\ell}}$ on the level set $\left\{ g=t\right\} $.
Since $F_{\sigma}$ can be expressed as an integral with respect to
the Barra-Gaspard measure, which is absolutely continuous with respect
to the Lebesgue measure, we see that $\mu_{\sigma}$ is as well.

This computation can be carried out assuming that $\nabla g$ does
not vanish identically on some open subset of a level set. Since the
function $g$ is real analytic, this can happen only if $g$ is constant
on a whole connected component of $\sreg$. Hence, absolute continuity
holds assuming that $g$ is not constant on a connected component
of $\sreg$. Nevertheless, the assumption that $g$ is not constant
on a connected component of $\sreg$ is not proven here. This assumption
is based on the intuition that such a constraint on $g$ (and in turn
on the values of $\vec{a}\left(\vec{\kappa}\right)$) is too restrictive
and on numerical experiments (see for example Figure \ref{fig:RNG-hist-1}
which suggests that $\mu_{\sigma}$ is indeed absolutely continuous).
Finally, we refer the reader to a conjecture of a similar spirit in
\cite[rem. 7.2]{AloBan_21ahp}.
\end{rem}

\begin{proof}[Proof of Theorem \ref{thm:converging_subsequence}]
 This proof once again uses the function $g\left(\vec{\kappa}\right)$
from the proof of Theorem \ref{thm:probability}.

Assume first that the graph's edge lengths are linearly independent
over $\mathbb{Q}$, and so the torus flow considered in Section \ref{subsec:The-secular-manifold}
is dense. Recall that the intersections of the torus flow with $\sreg$
is exactly $\left(\vec{\kappa}_{n}\right)_{n=1}^{\infty}$. Since
$g$ is continuous on $\sreg$ and the torus flow is dense, then for
every value $c\in\overline{Im\left(g\right)}$, there exists a subsequence
$g\left(\vec{\kappa}_{n_{m}}\right)$ along the torus flow such that
$g\left(\vec{\kappa}_{n_{m}}\right)\rightarrow c$. In order to connect
the RNG with the values of $g$, we proceed exactly as in the proof
of Theorem \ref{thm:probability}. Namely, combining (\ref{eq: RNG as square value of eigenfunction})
within the proof of Theorem \ref{thm:mean_RNG} together with (\ref{eq: sqaure of eigenfunction on secular manifold})
in the proof of Theorem \ref{thm:Weyl-law} gives that
\begin{equation}
\lim_{m\rightarrow\infty}d_{n_{m}}\left(\sigma\right)=\lim_{m\rightarrow\infty}\sigma\cdot g\left(\vec{\kappa}_{n_{m}}\right)=c\sigma,\label{eq:}
\end{equation}
and that this convergence is uniform in $\sigma$. These are the only
possible accumulation points for the RNG, since the arguments above
also show that any accumulation point must belong to $\overline{\left\{ \sigma\cdot g\left(\vec{\kappa}_{n}\right)\right\} _{n=1}^{\infty}}=\overline{\sigma\cdot g\left(\sreg\right)}$.
Since $g$ is a bounded and continuous function on a space with finitely
many connected components, $\overline{g\left(\sreg\right)}$ is a
finite collection of closed intervals. These intervals are of course
independent of the edge lengths of the graph (since $g$ was defined
independently of them). This proves the statement under the assumption
above (edge lengths are linearly independent over $\Q$).

Similar to the proof of Theorem \ref{thm:probability}, the rationally
dependent case can be proven by projecting the torus flow to the appropriate
subtorus where the flow is dense, and repeating the same argument.
Note that unlike the rationally independent case, the possible limit
points now depend on the subtorus to which we restrict the function
$g$, which of course depends on the graph edge lengths.

Lastly, in the degenerate case where $\ell_{e}/\ell_{e'}\in\mathbb{Q},\forall e,e'\in\mathcal{E}$,
the mentioned subtorus is one-dimensional (with a periodic flow along
it), which implies that $\left(g\left(\vec{\kappa}_{n}\right)\right)_{n=1}^{\infty}$
attains only finitely many values. Hence, the set of accumulation
points is finite and these are the values $a_{i}=b_{i}$ in the statement
of the theorem.
\end{proof}
While Theorem \ref{thm:converging_subsequence} above shows the existence
of converging subsequences for the RNG, it does not concretely specify
the possible limit points. The proof of the theorem shows that that
these limit points are determined by the possible values of the observable
$\left|f_{n}\left(v\right)\right|^{2}$. We now discuss this briefly.

First, the discussion in Section \ref{subsec:Discussion-bounds} shows
that zero is always an accumulation point for the RNG. This was proven
for the case of a star graph in \cite{RivRoy_jphys20}. In the two-dimensional
setting, \cite{RudWigYes_arxiv21} shows that under the assumption
that the billiard dynamics associated with the domain is ergodic,
there exists a subsequence of density $1$ which converges pointwise
to $\left\langle d\right\rangle _{n}\left(\sigma\right)$. The example
of an equilateral star graph (Figure \ref{fig:equistar-gap}) shows
that it is not true in general for graphs. 
\begin{figure}
\includegraphics[scale=0.4]{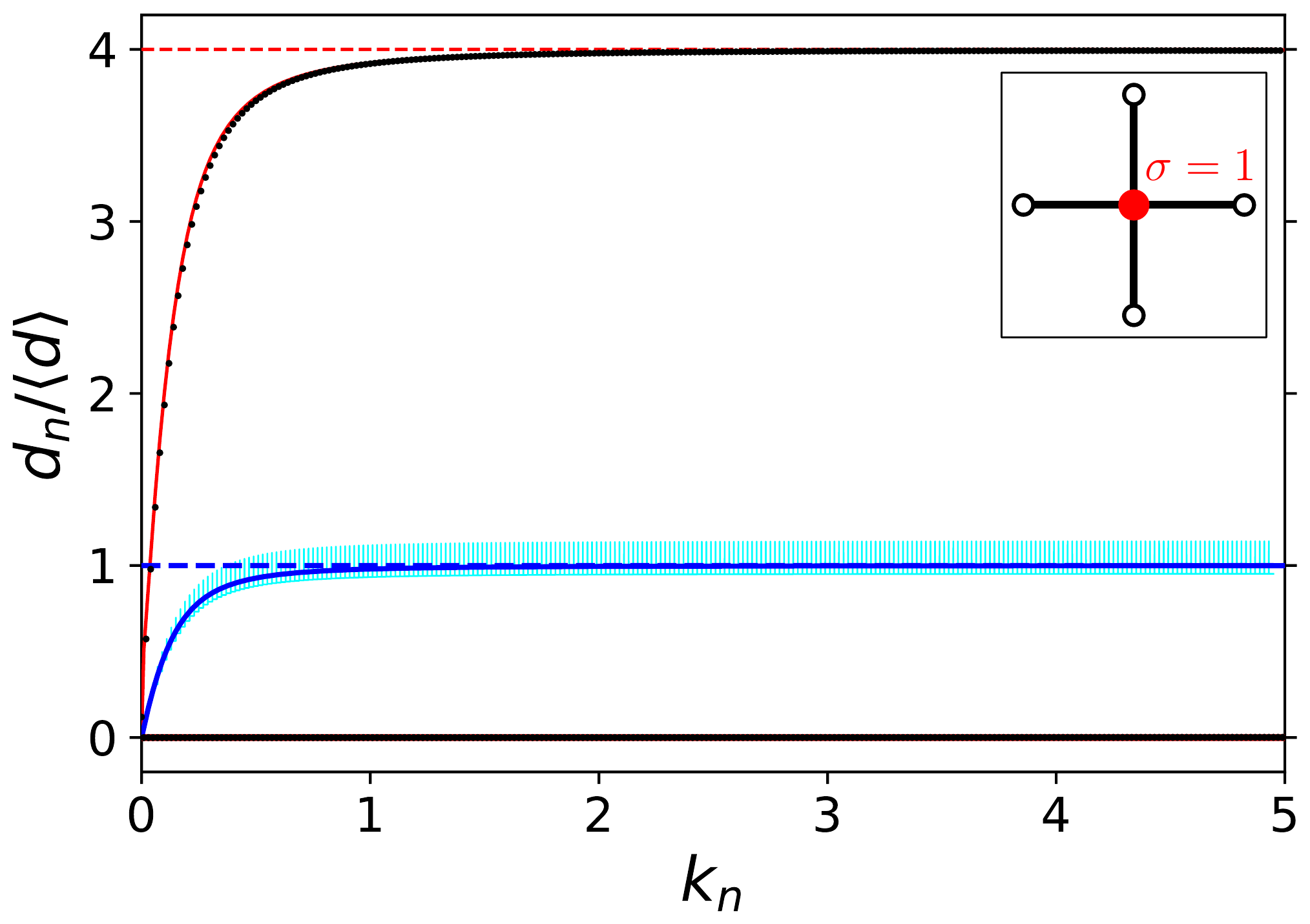}

\caption{\label{fig:equistar-gap} RNG (black points) for an equilateral star
graph with Robin condition at the central vertex, scaled so that $\langle d\rangle_{n}=1$
.The fluctuating light blue line is a running average and the blue
lines on top of it are the analytic results from Equations (\ref{eq:RNG_mean_const}),
(\ref{eq:RNG_mean_arctan}). The red dashed and full lines are the
analytic bounds of Equation (\ref{eq:explicit_bounds}), (\ref{eq:upperbound}).
Since many of the states vanish at the central vertex, their corresponding
RNG is zero. The RNG accumulates at two particular values, and does
not get close to the mean value.}
\end{figure}
 We believe that under the assumption that the graph's edge lengths
are linearly independent over $\mathbb{Q}$ (which is analogous to
the ergodic billiards assumption made in \cite{RudWigYes_arxiv21}),
then there does exist a subsequence which converges to the mean value.
This is equivalent to stating that the support of the measure $\mu_{\sigma}$
from Theorem \ref{thm:probability} contains the point $\left\langle d\right\rangle _{n}\left(\sigma\right)$.
This was proven for a star graph in \cite{RivRoy_jphys20}, and we
can also prove this for graphs which contain loops, but the general
case is still unclear. Even if such a subsequence exists, we believe
that it cannot be of density $1$ as in the two-dimensional setting.
Indeed, if the measure $\mu_{\sigma}$ is absolutely continuous in
general (as we conjecture in Remark \ref{rem:Absolutely-continuous-conjecture})
then there is no converging subsequence of density $1$.

\bigskip

\section{Proof and Discussion of Theorem \ref{thm:explicit_bounds} \label{sec:bounds proof}}

\subsection{Proof of Theorem \ref{thm:explicit_bounds}}

The main tool in the proof of Theorem \ref{thm:explicit_bounds} is
the following lemma, which provides an upper bound for the eigenvalue
derivatives with respect to the Robin parameter. This is done in terms
of the star decomposition introduced before Theorem \ref{thm:explicit_bounds}.
Adding to those notations, we also denote
\begin{equation}
s_{v}:=\frac{1}{\sum_{u\in\Uv}s_{v,u}^{-1}},\label{def:sv}
\end{equation}
so that $s_{v}$ is the harmonic mean of the edge lengths of the star
$\mathcal{S}_{v}$, divided by the number of edges in $\Sv$. If an
auxiliary vertex coincides with a graph vertex such that $s_{v,u}=0$,
then $s_{v}=0$.
\begin{lem}
\label{lemma:explicit_bounds} For any star decomposition of $\Gamma$
and any eigenvalue $\lambda_{n}\left(\sigma\right)>0$ ,
\begin{eqnarray}
0\leq\frac{\rmd\lambda_{n}\left(\sigma\right)}{\rmd\sigma}\le2\max_{v\in\VR}\left(\left|\Sv\right|+\frac{\sigma^{2}s_{v}+\sigma}{\lambda_{n}\left(\sigma\right)}\right)^{-1}\,.\label{eq:explicit_bounds_sensi}
\end{eqnarray}
For the lowest eigenvalue $\lambda_{1}\left(0\right)=0$ at $\sigma=0$,
the following holds:
\begin{equation}
\frac{\rmd\lambda_{1}\left(\sigma\right)}{\rmd\sigma}\big|_{\sigma=0}=\frac{\left|\VR\right|}{\left|\Gamma\right|}.\label{eq:-29}
\end{equation}
\end{lem}

\begin{proof}
The lower bound holds trivially by Lemma \ref{lem:hadamard}:
\begin{equation}
\frac{\rmd\lambda_{n}\left(\sigma\right)}{\rmd\sigma}=\sum_{v\in\VR}|f_{n}^{\left(\sigma\right)}\left(v\right)|^{2}\geq0.\label{eq:hadamard2}
\end{equation}
For $\lambda_{1}\left(0\right)$, the corresponding eigenfunction
is constant, and normalization implies that $\left|f_{1}^{\left(0\right)}\left(v\right)\right|^{2}=\frac{1}{\left|\Gamma\right|}$
for all $v$, leading to (\ref{eq:-29}).

To prove the upper bound, we fix an eigenvalue $\lambda_{n}\left(\sigma\right)>0$,
and denote it by $\lambda$ for brevity. On an edge $e=\left\{ v,v'\right\} $,
the corresponding eigenfunction can be written as 
\begin{equation}
\left.f\right|_{e}\left(x_{e}\right)=A_{e}\cos\left(kx_{e}-\varphi_{e,v}\right),\label{eq:cosine}
\end{equation}
where $x_{e}$ is the distance from $v$. Alternatively, using $x_{e}':=\ell_{e}-x_{e}$
and $\varphi_{e,v'}:=k\ell_{e}-\varphi_{e,v}$, one can write
\begin{equation}
\left.f\right|_{e}\left(x_{e}'\right)=A_{e}\cos\left(kx_{e}'-\varphi_{e,v'}\right).\label{eq:cosine-1}
\end{equation}
The following relations then hold: 
\begin{eqnarray}
f\left(v\right)=A_{e}\cos\varphi_{e,v} & \qquad & \frac{\rmd\left.f\right|_{e}}{\rmd x_{e}}\Big|_{x_{e}=0}=kA_{e}\sin\varphi_{e,v}\label{eq:A2u}\\
f\left(v'\right)=A_{e}\cos\varphi_{e,v'} &  & \frac{\rmd\left.f\right|_{e}}{\rmd x_{e}'}\Big|_{x_{e}'=0}=kA_{e}\sin\varphi_{e,v'},
\end{eqnarray}
and straightforward integration yields 
\begin{eqnarray}
\int_{0}^{\ell_{e}}(\left.f\right|_{e})^{2}\left(x\right)\rmd x_{e} & = & \frac{A_{e}^{2}}{2}\ell_{e}+A_{e}^{2}\frac{\sin\varphi_{e,v}\cos\varphi_{e,v}+\sin\varphi_{e,v'}\cos\varphi_{e,v'}}{2k}.\label{eq:edgeint}
\end{eqnarray}

Using $\left\Vert f\right\Vert _{L^{2}}=1$ and denoting the set of
edges connected to $v$ by $\mathcal{E}_{v}$ as in Subsection \ref{subsec:Basic-definitions},
we have that
\begin{eqnarray}
1 & = & \sum_{e\in\E}\int_{0}^{\ell_{e}}(\left.f\right|_{e})^{2}\left(x_{e}\right)\rmd x_{e}\label{eq:norm1}\\
 & = & \sum_{e\in\E}\frac{A_{e}^{2}}{2}\ell_{e}+\sum_{e\in\E}A_{e}^{2}\frac{\sin\varphi_{e,v}\cos\varphi_{e,v}+\sin\varphi_{e,v'}\cos\varphi_{e,v'}}{2k}\\
 & = & \sum_{e\in\E}\frac{1}{2}A_{e}^{2}\left(s_{v,u_{e}}+s_{v',u_{e}}\right)+\sum_{v\in\V}\sum_{e\in\E_{v}}A_{e}^{2}\frac{\sin\varphi_{e,v}\cos\varphi_{e,v}}{2k}\label{eq:norm2}\\
 & = & \sum_{v\in\V}\sum_{e\in\E_{v}}\frac{1}{2}A_{e}^{2}s_{v,u_{e}}+\sum_{v\in\V}\sum_{e\in\E_{v}}A_{e}^{2}\frac{\sin\varphi_{e,v}\cos\varphi_{e,v}}{2k},
\end{eqnarray}
where in the first term of (\ref{eq:norm2}) we used $\ell_{e}=s_{v,u_{e}}+s_{v',u_{e}}$,
which holds for any position of the auxiliary vertex $u_{e}$. We
would like to replace the amplitudes $A_{e}$ with the expressions
from (\ref{eq:A2u}). This can be done whenever $\cos\varphi_{e,v}\neq0$.
Since the case $\cos\varphi_{e,v}=0$ implies $f\left(v\right)=0$,
then dropping the (non-negative) terms with $f\left(v\right)=0$ gives
the following inequality:

\begin{eqnarray}
 & 1\ge & \sum_{v\in\V:f\left(v\right)\ne0}\frac{\left|f\left(v\right)\right|^{2}}{2}\sum_{e\in\E_{v}}\frac{s_{v,u_{e}}}{\cos^{2}\varphi_{e,v}}+\sum_{v\in\V}\frac{f\left(v\right)}{2k^{2}}\sum_{e\in\E_{v}}\frac{\rmd f_{e}}{\rmd x_{e}}\Big|_{x_{e}=0}\label{eq:norm3}\\
 & = & \sum_{v\in\VR}\frac{\left|f\left(v\right)\right|^{2}}{2}\left(\left|\mathcal{S}_{v}\right|+\sum_{e\in\E_{v}}s_{v,u_{e}}\tan^{2}\varphi_{e,v}+\frac{\sigma}{k^{2}}\right),\label{eq:norm4}
\end{eqnarray}
where we have used the Robin condition (\ref{eq:-18-1}) at $v$,
dropped the (non-negative) terms corresponding to $v\notin\VR$, and
replaced $\cos^{-2}\varphi_{e,v}=1+\tan^{2}\varphi_{e,v}$.

We wish to estimate the second term in Equation (\ref{eq:norm4}).
For $f\left(v\right)\ne0$, we substitute (\ref{eq:A2u}) into the
Robin condition (\ref{eq:-18-1}) and find that
\begin{eqnarray}
\forall v:f\left(v\right)\ne0,\quad\sum_{e\in\E_{v}}\tan\varphi_{e,v}=\frac{\sigma}{k}.\label{eq:sumtan}
\end{eqnarray}
Optimizing the second term in (\ref{eq:norm4}) in the variables $\varphi_{e,v}$
under the constraint (\ref{eq:sumtan}), one obtains the following
lower bound: 
\begin{eqnarray}
\sum_{e\in\E_{v}}s_{v,u_{e}}\tan^{2}\varphi_{e,v}\ge\frac{\sigma^{2}}{k^{2}}s_{v},\label{eq: optimizing sum of tan squares}
\end{eqnarray}
where $s_{v}$ is defined in (\ref{def:sv}). The lower bound in (\ref{eq: optimizing sum of tan squares})
is attained by
\begin{eqnarray}
\forall e\in\E_{v},\quad\tan\varphi_{v,e}=\frac{1}{s_{v,u_{e}}}\cdot\frac{\sigma}{k}s_{v}.\label{eq:mincond}
\end{eqnarray}
Substituting (\ref{eq: optimizing sum of tan squares}) in (\ref{eq:norm4})
we get:
\begin{eqnarray}
1 & \ge & \sum_{v\in\VR}\frac{\left|f\left(v\right)\right|^{2}}{2}\left(\left|\mathcal{S}_{v}\right|+\frac{\sigma^{2}s_{v}+\sigma}{k^{2}}\right)\label{eq:norm5}\\
 & \ge & \min_{v\in\VR}\left(\left|\mathcal{S}_{v}\right|+\frac{\sigma^{2}s_{v}+\sigma}{k^{2}}\right)\sum_{v\in\VR}\frac{\left|f\left(v\right)\right|^{2}}{2}.\label{eq:norm6}
\end{eqnarray}
Substituting this in (\ref{eq:hadamard2}) provides the required upper
bound in the lemma.
\end{proof}
The proof of Theorem \ref{thm:explicit_bounds} is now straightforward.
\begin{proof}
Integrating over the bounds from (\ref{eq:explicit_bounds_sensi})
gives 
\begin{eqnarray}
 &  & 0\leq d_{n}\left(\sigma\right)\le\int_{0}^{\sigma}2\max_{v\in\VR}\left(\left|\Sv\right|+\frac{s_{v}t^{2}+t}{\lambda_{n}\left(t\right)}\right)^{-1}\rmd t.\label{eq:fine-bound}
\end{eqnarray}
 To get the bounds (\ref{eq:explicit_bounds}) in the theorem, we
drop the second term in the expression above (which is positive).
The inequality (\ref{eq:-13}) follows since there always exists a
star decomposition which contains only stars whose total length is
at least $\frac{\ell_{\min}}{2}$; This star decomposition is obtained
by taking $u_{e}$ as the middle point of $e$ for all edges $e\in\E$.
\end{proof}
\bigskip

\subsection{Discussion of Lemma \ref{lemma:explicit_bounds} and Theorem \ref{thm:explicit_bounds}
\label{subsec:Discussion-bounds}}

\subsubsection{Optimality of the lower bound in Lemma \ref{lemma:explicit_bounds}
and Theorem \ref{thm:explicit_bounds}}

\LyXZeroWidthSpace{}

Under certain assumptions, the lower bound of zero in Lemma \ref{lemma:explicit_bounds}
and Theorem \ref{thm:explicit_bounds} is optimal. This happens when
the corresponding graph allows for Robin eigenfunctions whose absolute
values at the set $\VR$ are arbitrarily small. This results in very
low sensitivity to the Robin condition (i.e., small value of $\frac{\rmd\lambda_{n}}{\rmd\sigma}$),
giving an arbitrarily small value to the RNG. We demonstrate this
for two typical cases:
\begin{enumerate}
\item \label{enu: graph with cycle} The graph contains a cycle.
\item \label{enu: graph is tree}The graph is a tree, where at least two
leaves (denoted by $v_{1},v_{2}$) are not contained in $\VR$.
\end{enumerate}
For the graphs above, eigenfunctions with low sensitivity to the Robin
condition exist. In the case where the edge lengths of the graph are
linearly dependent over $\mathbb{Q}$ (rationally dependent), constructing
such eigenfunctions is simple. In fact, one can construct eigenfunctions
which vanish on the set $\VR$. By Lemma \ref{lem:hadamard}, these
eigenfunctions have zero sensitivity to the Robin condition (i.e.,
$\frac{\rmd\lambda_{n}}{\rmd\sigma}=0$ for all $\sigma\in\R$), and
they thus give zero RNG -- $d_{n}\left(\sigma\right)=0$ for all
$\sigma>0$. We now point out the existence of these eigenfunctions
(see also Figures \ref{fig:0-sens-states} and \ref{fig:0-sensitivity}).

\begin{figure}
\centerline{ \includegraphics[width=0.4\textwidth]{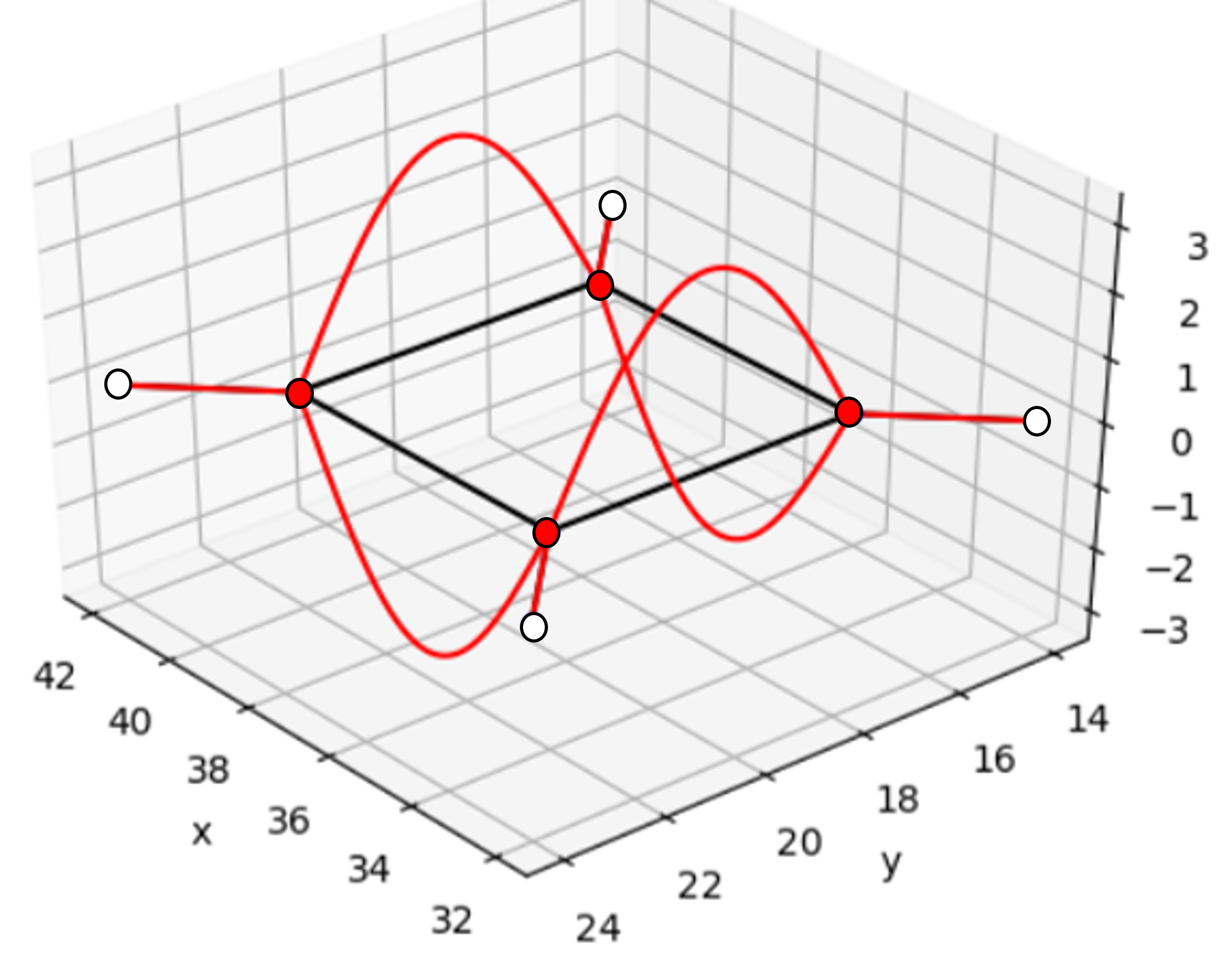}
\includegraphics[width=0.4\textwidth]{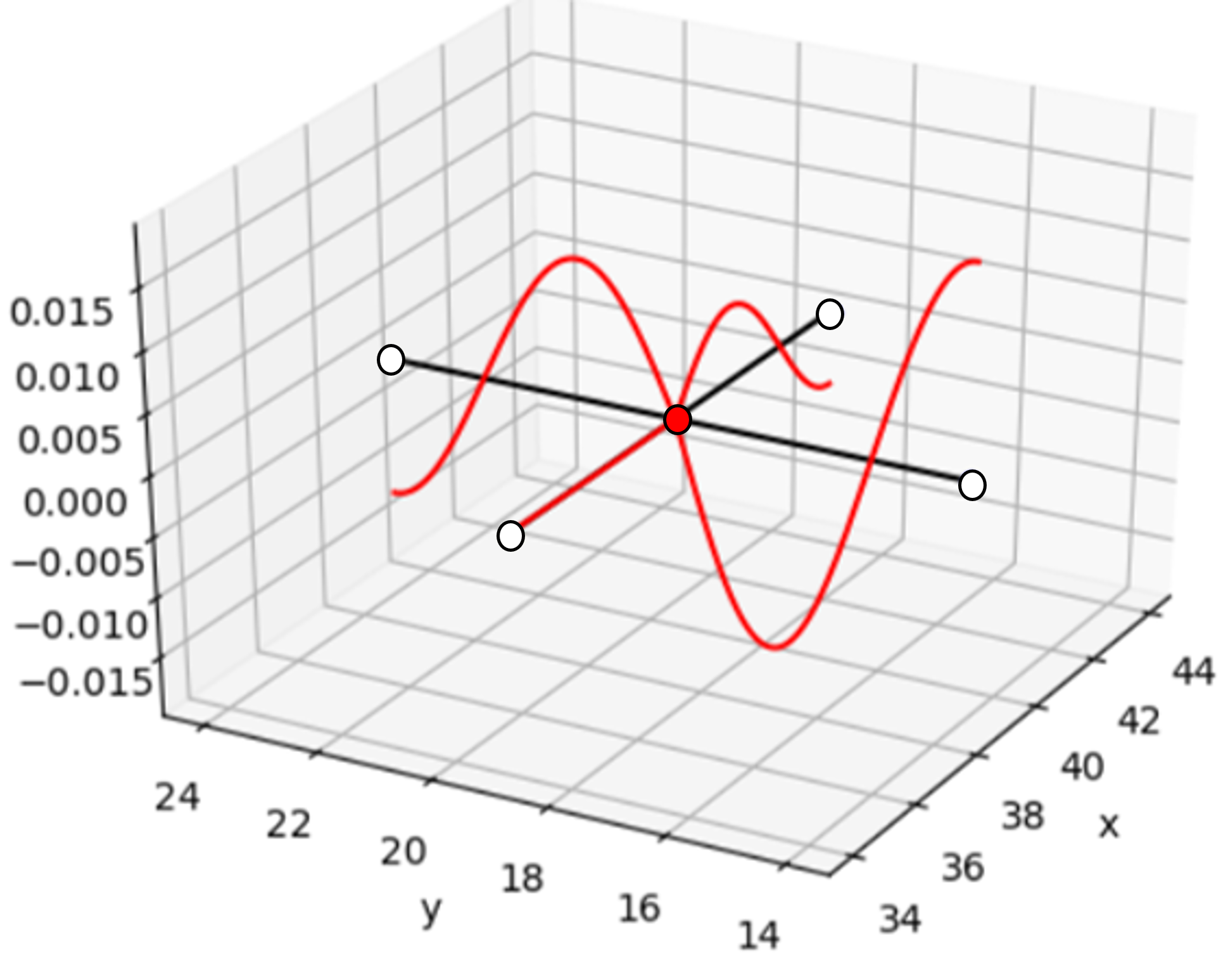} } \caption{\label{fig:0-sens-states} Two examples of states with zero sensitivity:
Case (\ref{enu: graph with cycle}) of a graph with a cycle, where
the Robin vertices are placed at the four vertices of the inner square.
Case (\ref{enu: graph is tree}) of a star graph with Robin condition
at the central vertex. In both cases the eigenfunction vanishes at
the set $\protect\VR$, resulting in zero sensitivity, i.e., $\frac{\protect\rmd\lambda_{n}}{\protect\rmd\sigma}$=0
for all $\sigma\in\protect\R$.}
\end{figure}

In case (\ref{enu: graph with cycle}), one can select a cycle on
the graph, and choose the wave number $k$ such that all edge lengths
in the cycle are integer multiples of the wave length $\Lambda:=2\pi/k$.
Under these conditions, there exist eigenfunctions which vanish on
the entire graph apart from the given cycle. In particular, those
eigenfunctions vanish at all vertices on the given cycle. In case
(\ref{enu: graph is tree}), consider the (unique) path connecting
the vertices $v_{1},v_{2}$. Then similar to before, one can choose
edge lengths and $k$ such that all edges in the path are integer
multiples of $\Lambda$, with the exception of the edges adjacent
to $v_{1},v_{2}$, to which an additional $\Lambda/4$ is added. Under
these conditions, there exist eigenfunctions which vanish on all of
the graph apart from the given path. In particular, those eigenfunctions
vanish at all interior vertices along the given path, and their derivative
at $v_{1},v_{2}$ vanishes. All eigenfunctions described above (for
both cases (\ref{enu: graph with cycle}) and (\ref{enu: graph is tree}))
vanish at $\VR$, as required.

\begin{figure}
\centerline{ \includegraphics[width=0.45\textwidth]{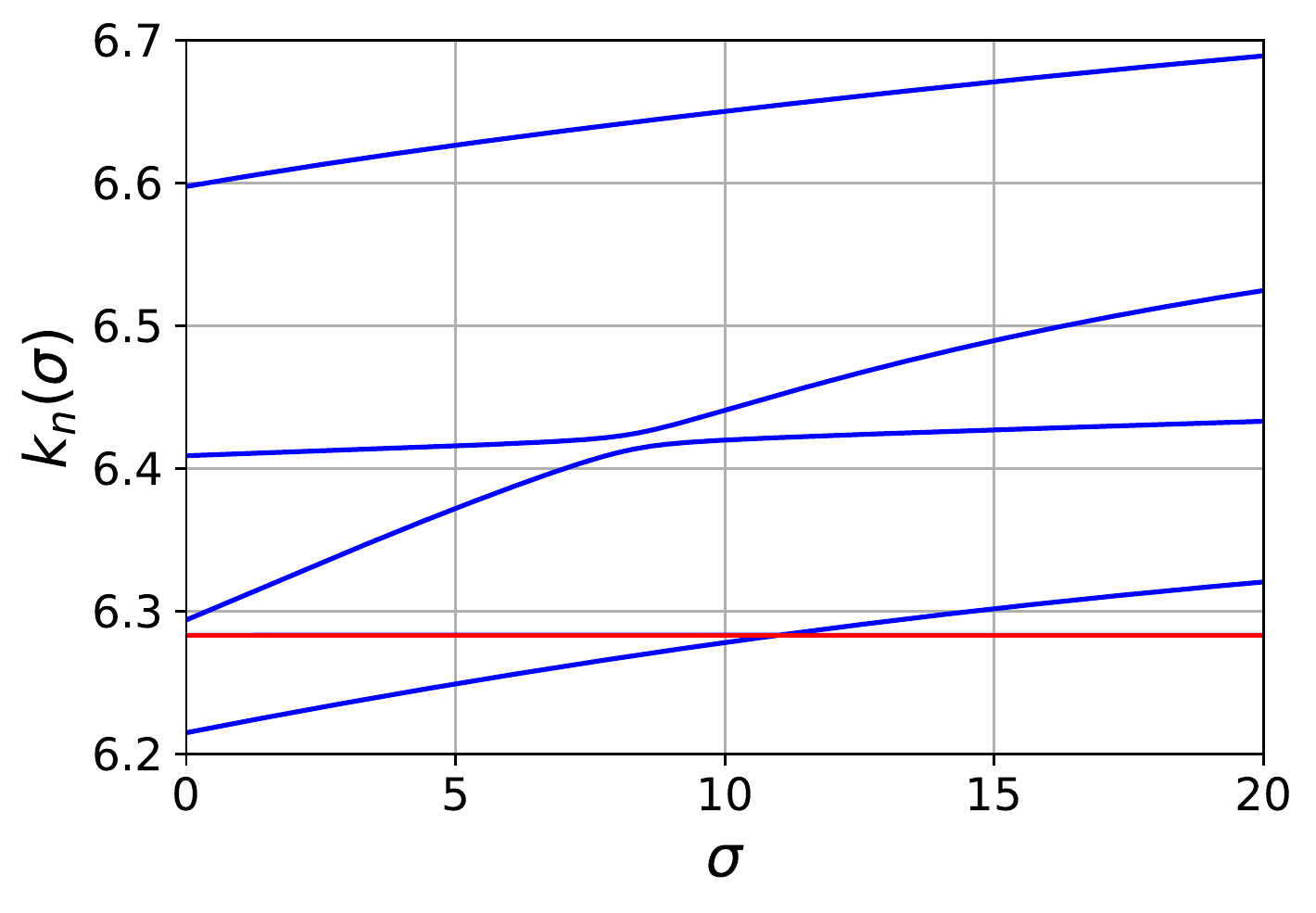} } \caption{\label{fig:0-sensitivity} The eigenvalue curves for a tetrahedron
graph with $\protect\VR=\protect\V$. The horizontal red curve at
$k=2\pi$ corresponds to a state which is supported on a triangle
subgraph and vanishes at all $v\in\protect\VR$. This gives a state
with zero sensitivity as in case (i), which results in $d_{n}\left(\sigma\right)=0$.}
\end{figure}

While it is impossible to construct eigenfunctions with zero sensitivity
in the general case of rationally independent edge lengths, one can
still use the construction above to find eigenfunctions with arbitrarily
small sensitivity. This can be done by approximating the given edge
lengths with rationally dependent edge lengths, and then applying
the method above. This will give a sequence of eigenfunctions whose
value at $\VR$ is arbitrarily small, resulting in a subsequence of
$d_{n}\left(\sigma\right)$ which tends to zero.

The frequency of such eigenfunctions with low sensitivity has been
estimated in \cite{Schanz2003}. It depends on the number of edge
lengths that determine the value of $k$ in the construction above.
For example, these eigenfunctions appear more frequently in Figure
\ref{fig:sensi} (a) than in (b). This is since in case (a) of a star
graph the supporting path which determines $k$ contains two edges,
while in case (b) of a tetrahedron it contains three edges forming
a cycle. The figures show that while the lower bound is not attained
in these cases, it is still tight.

\begin{figure}
\centerline{ \includegraphics[width=0.45\textwidth]{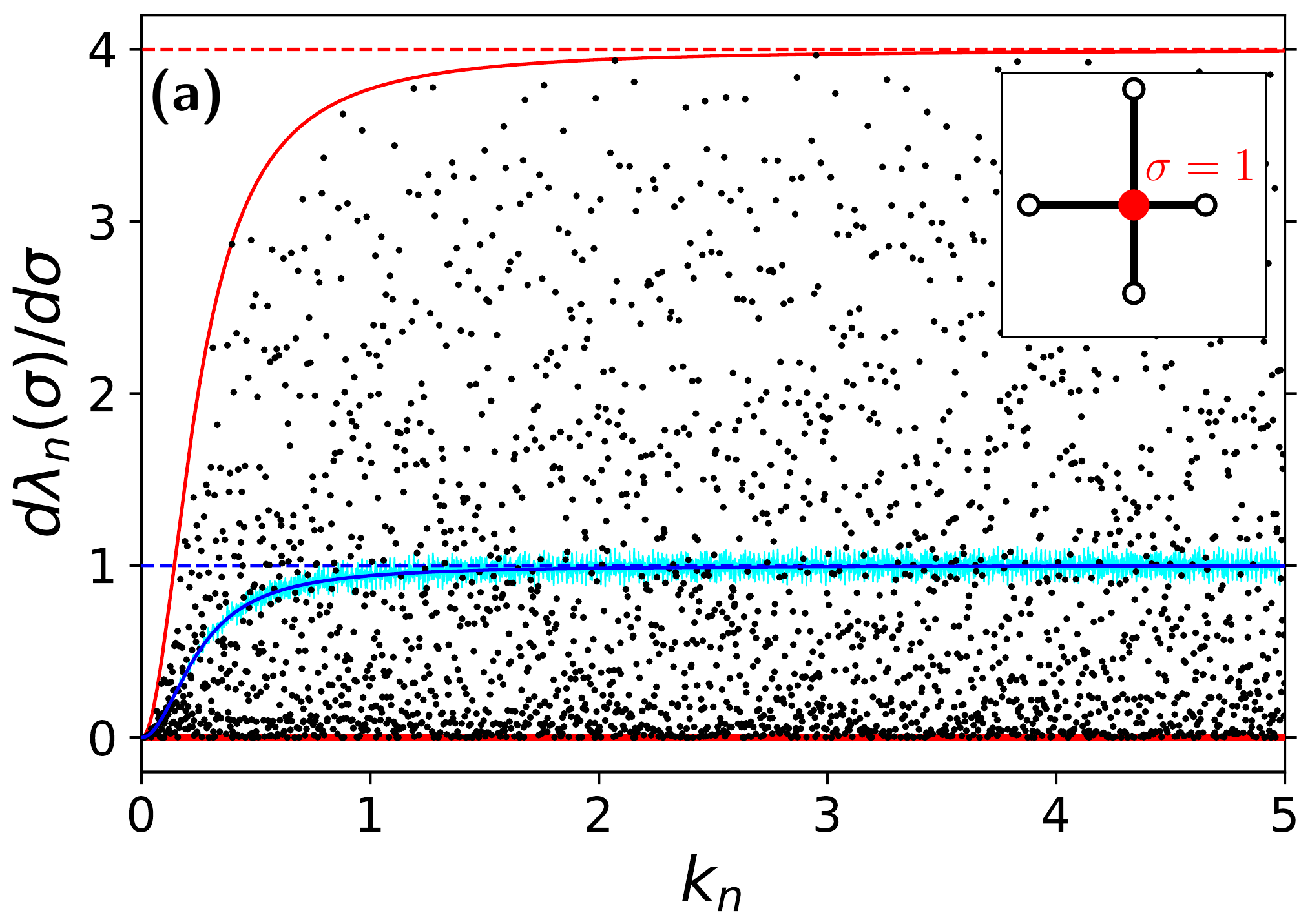} \includegraphics[width=0.45\textwidth]{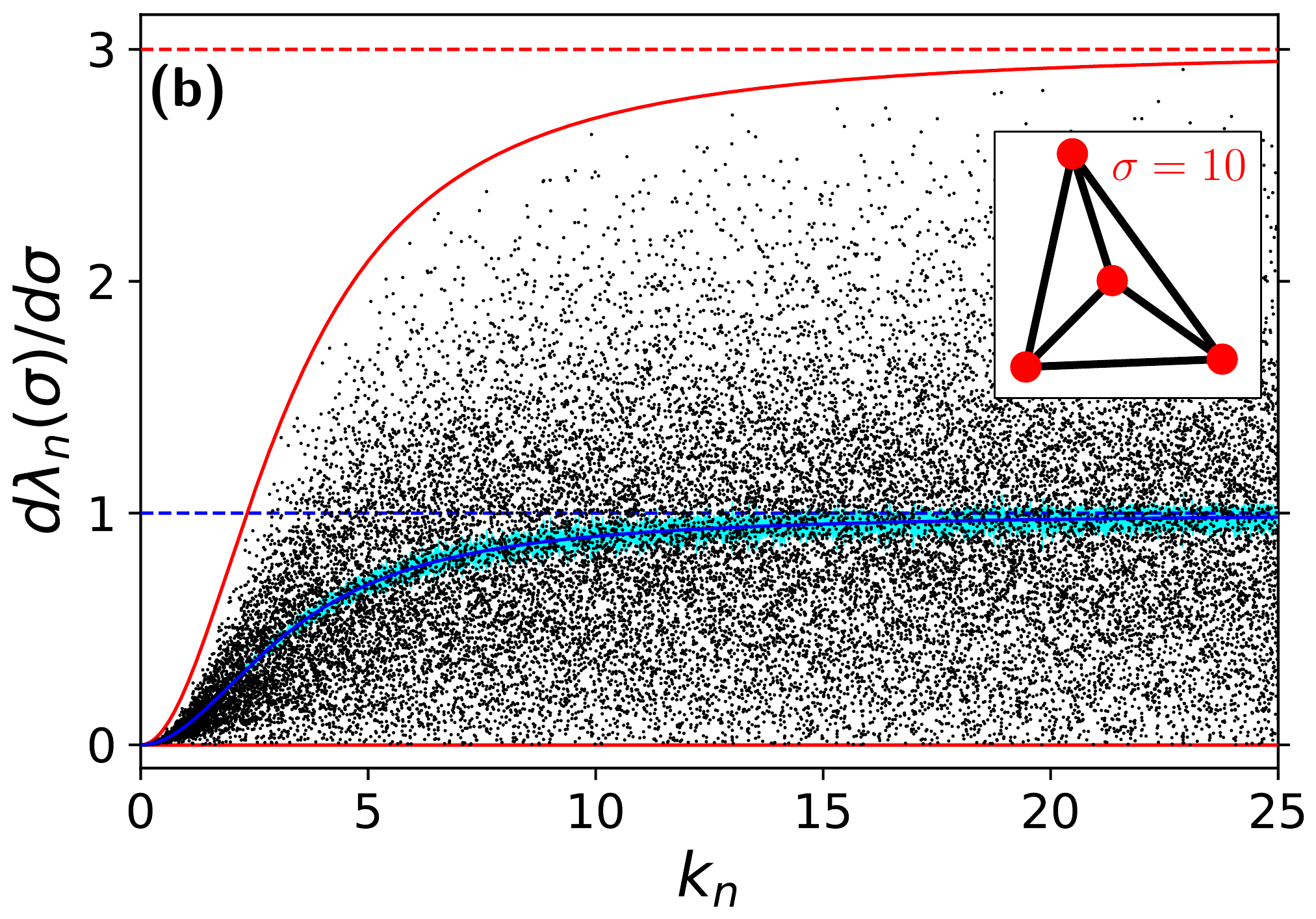}
} \caption{\label{fig:sensi} Scatter plot of the eigenvalue sensitivities $\frac{d\lambda_{n}}{d\sigma}$
for (a) the star graph from Figure \ref{fig:star_gap} and (b) the
tetrahedron graph from Figure \ref{fig:cg-gap} below. The values
are scaled so that the mean sensitivity is one. The light blue lines
are running averages, and the blue lines on top of it are the analytic
results from Equations (\ref{eq:-27}),(\ref{eq:-28}). The red lines
are the upper bound of (\ref{eq:explicit_bounds_sensi}) with and
without the second term, respectively. }
\end{figure}

\bigskip

\subsubsection{Optimality of the upper bound in Lemma \ref{lemma:explicit_bounds}}

\LyXZeroWidthSpace{}

In order to construct eigenfunctions which attain the upper bound
of Lemma \ref{lemma:explicit_bounds}, equality must hold in Equations
(\ref{eq:norm3}), (\ref{eq:mincond}), (\ref{eq:norm5}) and (\ref{eq:norm6}).
We provide two examples for specific graphs which satisfy this.

For the first example, consider an equilateral star graph with edges
of length $\ell$ and Robin condition at the central vertex $v_{0}$.
On each edge, the $L^{2}$ normalized eigenfunctions can be written
as
\begin{equation}
f_{e}\left(x_{e}\right)=\frac{1}{\sqrt{\deg(v_{0})\ell}}\cos\left(kx_{e}-\varphi_{e,v_{0}}\right).\label{eq:-30}
\end{equation}
Then by choosing $k>0$ such that 
\begin{equation}
\tan k\ell=\frac{\sigma}{k\deg(v_{0})},\label{eq:condmax}
\end{equation}
and taking $\varphi_{e,v_{0}}=k\ell$, we get a valid Robin eigenfunction.
Moreover, this eigenfunction does not vanish at any vertex, so equality
holds in (\ref{eq:norm3}). Now, choose a trivial star decomposition,
where the points $u_{e}$ are located at the outer vertices. In this
case, all terms in (\ref{eq:norm3}) which correspond to the outer
vertices vanish. For the remaining vertex $v_{0}$, (\ref{eq:mincond})
is satisfied and equality holds in $(\ref{eq:norm5})$. This also
gives equality in (\ref{eq:norm6}).

The assumption of equal edge lengths which was used for this construction
can be somewhat relaxed. Similar to the previous subsection, choosing
arbitrary edge lengths will lead to the existence of eigenfunctions
with sensitivity arbitrarily close to the upper bound. However, since
all edges of the graph are involved in the construction, the probability
of such an eigenfunction is much lower than in the case of approaching
the lower bound zero (where only a small subset of the edges were
involved), see Figure \ref{fig:sensi}. Moreover, condition (\ref{eq:condmax})
which is used to determine $k$ depends on $\sigma$ and can be satisfied
at most at isolated points along the integral in Equation (\ref{eq:explicit_bounds}).
Therefore, unlike the lower bound zero, the upper bound in Theorem
\ref{thm:explicit_bounds} cannot be realized by this construction.

The second construction is similar, but concerns a regular graph with
no neighboring Neumann vertices. For this construction, one should
take the length of edges connecting a Robin vertex to a Neumann vertex
to be $\ell$, and the length of an edge connecting two Robin vertices
to be $2\ell$. Just as before, choose $k$ according to (\ref{eq:condmax})
above and $\varphi_{e,v}=k\ell$ for all $v\in\V$. This will once
again give a valid eigenfunction. Now, choose a star decomposition
which is attained by splitting the edges which connect two Robin vertices
in the middle. i.e, $s_{v,u_{e}}=\ell$ for all Robin vertices and
$s_{v',u_{e}}=0$ for all Neumann vertices. Then just as in the example
above, the Neumann vertices have zero contribution to (\ref{eq:norm3}),
while for the Robin vertices (\ref{eq:mincond}) holds and so there
is equality in (\ref{eq:norm5}). Since all star graphs around Robin
vertices are identical, equality holds in (\ref{eq:norm6}) as well.

\bigskip

\subsubsection{Improved upper bound for the RNG \label{subsec:Explicit-estimate-of}}

\LyXZeroWidthSpace{}

We wish to give here a better bound than the one given in Theorem
\ref{thm:explicit_bounds}.

Fix $\sigma>0$. For every $t\in\left[0,\sigma\right]$ choose a star
decomposition, and denote by $v_{t}\in\VR$ the vertex which is selected
by the $\max$-condition in (\ref{eq:explicit_bounds_sensi}), 
\begin{equation}
\max_{v\in\VR}\left(\left|\Sv\right|+\frac{t^{2}s_{v}+t}{\lambda_{n}\left(t\right)}\right)^{-1}.\label{eq:max-cond}
\end{equation}
Fix the parameters $\check{s},\check{S}$ such that 
\begin{equation}
\forall t\in\left[0,\sigma\right],\quad\left|\mathcal{S}_{v_{t}}\right|\ge\check{S},\qquad s_{v_{t}}\ge\check{s}.\label{eq:constestimate}
\end{equation}

For the purpose of getting a finer upper bound on the RNG, we will
be interested to take the highest possible values for $\check{s},\check{S}$.
In particular, the bound in the next proposition will be valid only
if $\check{s},\check{S}$ are positive. We demonstrate that this can
be achieved in the following two cases:
\begin{enumerate}
\item \label{enu: s-Assumption - star graph}For a star graph with Robin
condition at the central vertex. The auxiliary vertices may be chosen
to be the boundary vertices of the graph. Thus, the (degenerate) stars
around these vertices have zero length. Since the boundary vertices
are imposed with the Neumann-Kirchhoff condition, it is easy to see
that $v_{t}$ is not one of the boundary vertices. Hence, we may choose
$\check{S}$ to be the total length of the graph, and $\check{s}$
to be the harmonic mean of its edge lengths divided by the degree.
Indeed, both are non-zero.
\item \label{enu: s-Assumption - arbitrary graph} Let the graph be arbitrary,
with the same star decomposition fixed for all $t\in\left[0,\sigma\right]$.
Then one can take $\check{S}$ as the minimal length of a star around
a Robin vertex in the partition and $\check{s}$ as the minimum value
of $s_{v}$ for a star in the partition. If none of the partition
vertices are placed at $\VR$, then indeed both $\check{S}$ and $\check{s}$
are non-zero. In general, the bound stated in the next proposition
is tighter if $\check{S},\check{s}$ are large, i.e. if the stars
around the Robin vertices are large. Therefore, if there are edges
which connect a Robin vertex $u\in\VR$ with a Neumann vertex $w\in\V\backslash\VR$,
then to maximize the values of $\left|\Sv\right|$ and $s_{v}$, one
should choose a star partition for which the auxiliary vertex should
be placed at $w$, similarly to what was done in case (\ref{enu: s-Assumption - star graph}). 
\end{enumerate}
\begin{prop}
\label{prop:integral_estimate} For $\check{s},\check{S}$ as above
and $\lambda_{n}\left(0\right)>\frac{1}{4\check{s}\check{S}}$,
\begin{equation}
d_{n}\left(\sigma\right)<\left(\frac{\exp\left(2\alpha\arctan\left(\frac{\alpha}{2}\cdot\left[1+2\check{s}\sigma\right]\right)\right)}{\exp\left(2\alpha\arctan\left(\frac{\alpha}{2}\right)\right)}-1\right)\cdot\lambda_{n}\left(0\right),\label{eq:upperbound}
\end{equation}
with 
\begin{equation}
\alpha=\frac{2}{\sqrt{4\lambda_{n}\left(0\right)\check{s}\check{S}-1}}.\label{eq:alpha}
\end{equation}
\end{prop}

\begin{proof}
From Equation (\ref{eq:explicit_bounds_sensi}) in Lemma \ref{lemma:explicit_bounds},
it follows that 
\begin{eqnarray}
\frac{d\lambda_{n}\left(t\right)}{dt} & \leq & \frac{2\lambda_{n}\left(t\right)}{\left|\mathcal{S}_{v_{t}}\right|\lambda_{n}\left(t\right)+s_{v_{t}}t^{2}+t}\leq\frac{2\lambda_{n}\left(t\right)}{\check{S}\lambda_{n}\left(0\right)+\check{s}t^{2}+t}.\label{eq:sensestimate1}
\end{eqnarray}
Thus, $\lambda_{n}\left(\sigma\right)$ is bounded from above by the
function $\lambda\left(\sigma\right)$ satisfying the differential
equation 
\begin{eqnarray}
\frac{d\lambda}{dt}=\frac{2\lambda\left(t\right)}{A+2Bt+Ct^{2}},
\end{eqnarray}
with $A=\lambda_{n}\left(0\right)\check{S}$, $B=1/2$, $C=\check{s}$
and initial condition $\lambda\left(0\right)=\lambda_{n}\left(0\right)$.
Solution by separation of variables gives
\begin{equation}
\ln\frac{\lambda\left(\sigma\right)}{\lambda\left(0\right)}=\int_{0}^{\sigma}dt\frac{2}{A+2Bt+Ct^{2}}=2\alpha\arctan\left(\alpha\left[Ct+B\right]\right)\Big|_{0}^{\sigma}\qquad(AC>B^{2}),\label{eq:logdiff}
\end{equation}
with $\alpha=\left(AC-B^{2}\right)^{-1/2}$ corresponding to (\ref{eq:alpha}).

Solving (\ref{eq:logdiff}) above for $\lambda_{n}\left(\sigma\right)-\lambda_{n}\left(0\right)$
provides the given upper bound for $d_{n}\left(\sigma\right)$. Since
equality holds in (\ref{eq:sensestimate1}) only for $\lambda_{n}\left(t\right)=\lambda_{n}\left(0\right)$
(and then $d_{n}\equiv0$), the inequality in (\ref{eq:upperbound})
is strict.

\begin{figure}
\includegraphics[scale=0.5]{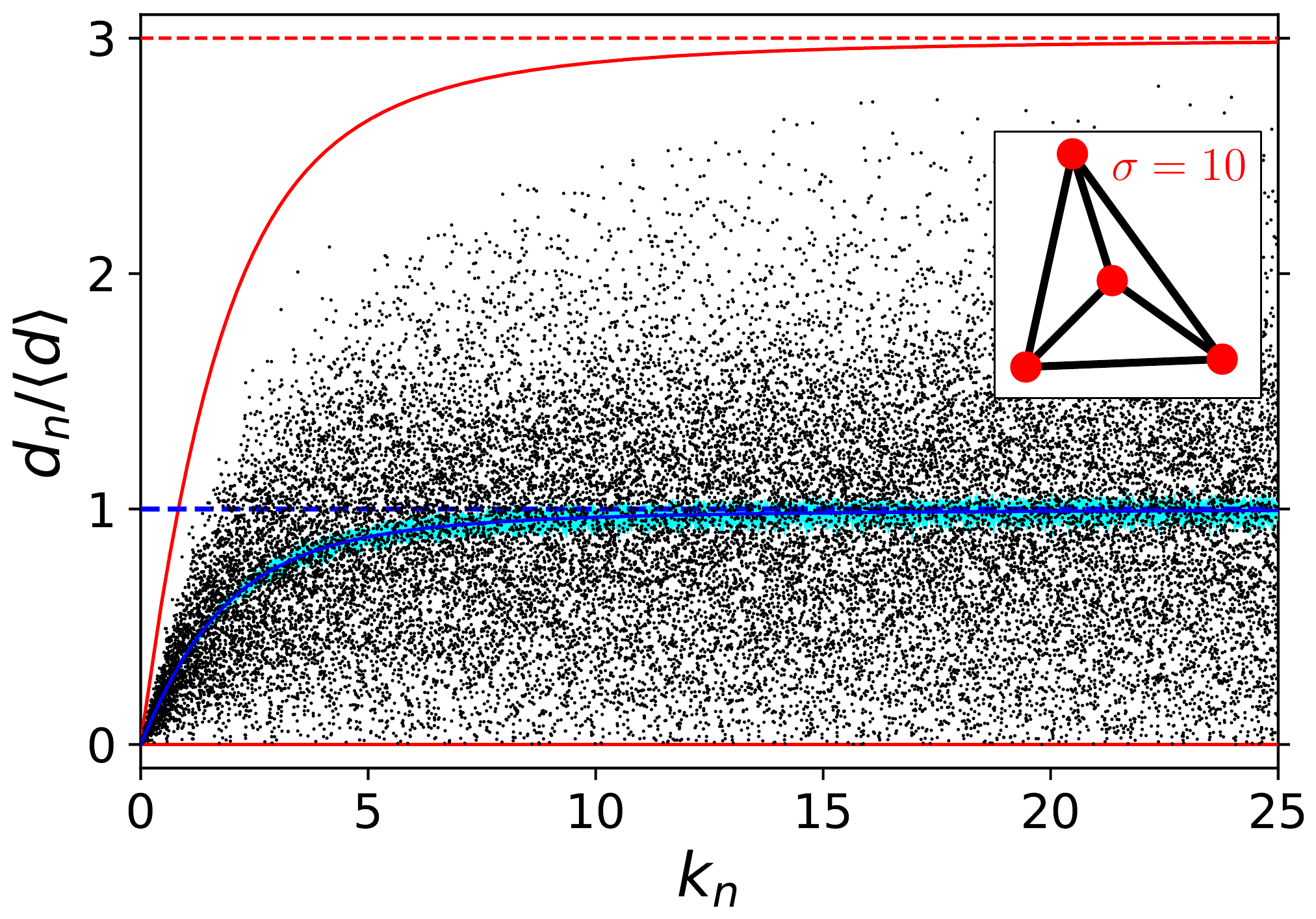}

\caption{Scatter plot of the first 25,000 Robin-Neumann gaps for a tetrahedron
with Robin condition at all vertices, scaled so that $\left\langle d\right\rangle _{n}\left(\sigma\right)=1$.
The light blue line is a running average and the blue lines on top
of it are the results from (\ref{eq:RNG_mean_const}), (\ref{eq:RNG_mean_arctan}).
The red dashed and full lines are the bounds from (\ref{eq:explicit_bounds}),
(\ref{eq:upperbound}). \label{fig:cg-gap}}
\label{fig:RNG of tetrahedron}
\end{figure}
\end{proof}
Figures \ref{fig:star_gap} and \ref{fig:sensi}(a) demonstrate the
sensitivity and RNG values for a star graph with four edges and $\VR$
consisting of just the central vertex. In these figures, the solid
red line was calculated with parameters as described in case (i) above.
By Theorem \ref{thm:mean_RNG}, the limiting mean value of the RNG
for this star graph is $\sigma/2\left|\Gamma\right|$, while from
(\ref{eq:explicit_bounds}) with $\left|\Sv\right|=\left|\Gamma\right|$
we find that the upper bound is larger than the mean by a factor of
$\deg\left(v_{0}\right)=4$. Figures \ref{fig:sensi}(b) and \ref{fig:cg-gap}
demonstrate the sensitivity and RNG values for a tetrahedron graph
with $\VR=\V$. The optimal partition used to draw the bounds in these
figures was found numerically. It has $\left|\Sv\right|=\left|\Gamma\right|/4$,
i.e., all stars in the partition have the same total length. Once
again, the ratio between the mean gap and upper bound is determined
by the degree of the Robin vertices ($3$ in this case).

\section{An additional approach for deriving the mean Robin-Neumann gap \label{sec:discuss_mean_RNG}}

In this section, we present an alternative derivation for the mean
value of the Robin-Neumann gap. This is done by considering a so-called
``local average'' of the RNG with respect to the wave number $k$
(instead of averaging with respect to $n$). This approach is not
as rigorous as the proof of Theorem \ref{thm:mean_RNG}. Nevertheless,
it is advantageous in that it provides not only the limiting mean
value of the RNG, but also the running mean as it depends on $k$,
see Figure \ref{fig:star_gap}.

We begin by considering the situation where the Robin condition is
imposed at a single vertex, and later generalize to multiple vertices.
In this case, we know that the eigenvalues interlace (see \cite[thm. 3.1.8]{BerKuc_incol12});
if $\sigma<\sigma'$, then for all $n\in\mathbb{N}$
\begin{equation}
k_{n}\left(\sigma\right)\le k_{n}\left(\sigma'\right)<k_{n+1}\left(\sigma\right).\label{}
\end{equation}
This can be rewritten in terms of the number counting function $\ncf\left(k,\sigma\right):=\left|\left\{ n\in\mathbb{N}:k_{n}\le k\right\} \right|$.
Mainly, this means that the spectral shift, which is the difference
between the number counting functions at fixed $k$, may only take
the values zero and one,
\begin{equation}
\Delta^{\sigma}\ncf\left(k\right):=\ncf\left(k,0\right)-\ncf\left(k,\sigma\right)\in\left\{ 0,1\right\} .\label{}
\end{equation}
We denote the length of the intervals where the spectral shift is
equal to one by
\begin{equation}
\delta_{n}\left(\sigma\right):=k_{n}\left(\sigma\right)-k_{n}\left(0\right).\label{-1}
\end{equation}
Note that these intervals are defined similarly as the RNG, but for
the difference between the $k$ values rather than the eigenvalues.

By the Weyl law for metric graphs (see \cite{BerKuc_incol12,GnuSmi_ap06})
for a fixed value of $\sigma$, the mean distance between consecutive
values of $\left\{ k_{n}(\sigma)\right\} $ is 
\begin{eqnarray}
 & \deltak:=\pi/\left|\Gamma\right|.\label{eq:deltak}
\end{eqnarray}
Hence, for large $K>0$, the interval $\left[k-K/2,k+K/2\right]$
contains on average $N:=K/\deltak$ values from $\left\{ k_{n}(\sigma)\right\} $.
Thus, defining a local $k$-average, the spectral shift in $k$ is
equal to:
\begin{eqnarray}
\overline{\Delta^{\sigma}\ncf}\left(k\right) & = & \frac{1}{K}\int_{k-K/2}^{k+K/2}\Delta\ncf^{\sigma}\left(k'\right)dk'\approx\frac{\left|\Gamma\right|}{\pi}\cdot\frac{1}{N}\sum_{n=N_{0}+1}^{N_{0}+N}\delta_{n}\left(\sigma\right)=\frac{\left|\Gamma\right|\delta^{\sigma}\left(k\right)}{\pi},\label{eq:meanshift}
\end{eqnarray}
where $\{N_{0}+1,\ldots,N_{0}+N\}$ are the indices of the $\left\{ k_{n}(\sigma)\right\} $
values which are contained in the interval $\left[k-K/2,k+K/2\right]$
(on average). Hence, $\delta^{\sigma}\left(k\right):=\frac{1}{N}\sum_{n=N_{0}+1}^{N_{0}+N}\delta_{n}\left(\sigma\right)$
is the mean spectral shift around $k$. The expression above holds
up to an error of order $N^{-1}$ due to the limits of the integration
interval.

To evaluate the mean spectral shift above, we use the trace formula
for the counting function (as derived in \cite{KotSmi_prl97,KotSmi_ap99}):
\begin{equation}
\ncf^{\sigma}\left(k,\sigma\right)=\frac{\Theta\left(k,\sigma\right)}{2\pi}+\frac{1}{\pi}\cdot Im\sum_{m=1}^{\infty}\frac{\tr U^{m}\left(k,\sigma\right)}{m}.\label{eq:numbercounting}
\end{equation}
Here, $U\left(k,\sigma\right):=S^{\left(\sigma\right)}e^{ik\L}$ is
the unitary scattering matrix (as in Theorem \ref{thm:secular-equation})
and $\Theta\left(k,\sigma\right):=\log\left(\det\left(U\left(k,\sigma\right)\right)\right)$
is known as the total phase of $U\left(k,\sigma\right)$. Under the
assumption that $K$ is large enough (mainly, that $K\gg\pi/\ell_{\min}$,
see \cite{KotSmi_prl97,KotSmi_ap99}), the contribution of the oscillatory
term in (\ref{eq:numbercounting}) is suppressed by the averaging,
and in leading order we have that
\begin{equation}
\overline{\Delta^{\sigma}\ncf}\left(k\right)=\frac{\overline{\Theta\left(k,0\right)}-\overline{\Theta\left(k,\sigma\right)}}{2\pi}.\label{totalphase}
\end{equation}
The total phase was evaluated in \cite{KotSmi_ap99} as
\begin{equation}
\Theta\left(k,\sigma\right)=2k\left|\Gamma\right|-2\sum_{v\in\VR}\arctan\left(\frac{\sigma}{\deg\left(v\right)k}\right).\label{tot-phase}
\end{equation}
Plugging this into (\ref{totalphase}) and then using (\ref{eq:meanshift})
gives that for a single Robin vertex
\begin{eqnarray}
\delta^{\sigma}\left(k\right) & = & \frac{1}{\left|\Gamma\right|}\arctan\left(\frac{\sigma}{\deg\left(v\right)k}\right).\label{eq:avdeltak}
\end{eqnarray}
Finally, we can define the $k$-averaged Robin-Neumann gap by 
\begin{equation}
\left\langle d\right\rangle _{k}\left(\sigma\right):=\left(k+\delta^{\sigma}\left(k\right)\right)^{2}-k^{2},\label{eq:-26}
\end{equation}
which under the assumption $\delta^{\sigma}\left(k\right)\ll k$,
and together with (\ref{eq:avdeltak}), gives the following:

\begin{equation}
\left\langle d\right\rangle _{k}\left(\sigma\right)\approx2k\delta^{\sigma}\left(k\right)=\frac{2k}{\left|\Gamma\right|}\arctan\left(\frac{\sigma}{\deg\left(v\right)k}\right).\label{eq:-25}
\end{equation}

For the more general case where the Robin condition is imposed at
several vertices, we can repeat the same proof (applying the additivity
of Equation (\ref{tot-phase})) to obtain
\begin{equation}
\left\langle d\right\rangle _{k}\left(\sigma\right)=\frac{2k}{\left|\Gamma\right|}\sum_{v\in\V_{\mathcal{R}}}\arctan\left(\frac{\sigma}{\deg\left(v\right)k}\right).\label{eq:RNG_mean_arctan}
\end{equation}

Note that for $k\to\infty$, one can recover the rigorously obtained
expression from Theorem \ref{thm:mean_RNG} by first order approximation
of (\ref{eq:RNG_mean_arctan}). On the other hand, for $k\to0$ the
average gap approaches zero. The average sensitivity of the gaps with
respect to a change of the Robin parameter is obtained by differentiating
(\ref{eq:RNG_mean_arctan}) with respect to $\sigma$:
\begin{align}
 & \left\langle \frac{d\lambda\left(\sigma\right)}{d\sigma}\right\rangle _{n}=\frac{2}{\left|\Gamma\right|}\sum_{v\in\VR}\frac{\lambda\deg\left(v\right)}{\sigma^{2}+\lambda\deg\left(v\right)^{2}}\label{eq:-27}\\
 & \approx_{\lambda\rightarrow\infty}\frac{2}{\left|\Gamma\right|}\sum_{v\in\VR}\frac{1}{\deg\left(v\right)}.\label{eq:-28}
\end{align}

A crucial assumption in the derivations above was that the averaging
interval $K$ contains many eigenvalues, which is required to neglect
the oscillating terms in (\ref{eq:numbercounting}). At the same time,
$K$ must be small enough so that the value of (\ref{eq:avdeltak})
does not change by a large amount inside the given interval. Otherwise,
the definition of a local average of the gap is not meaningful. The
two conditions can only be met for graphs with a large metric length
$\left|\Gamma\right|\to\infty$. Alternatively, one may employ an
ensemble average over graphs where the topology is fixed and the edge
lengths are varied. Having written the above, we refer to Figures
\ref{fig:star_gap}, \ref{fig:sensi} and \ref{fig:RNG of tetrahedron},
which demonstrate how close is (\ref{eq:RNG_mean_arctan}) to a running
mean value obtained by averaging over $21$ adjacent eigenvalues.

\bigskip

\section{Discussion and open questions\label{sec:discussion}}

We conclude this work by comparing its results to the ones presented
in \cite{RivRoy_jphys20,RudWig_amq21,RudWigYes_arxiv21} and raising
several open questions.

\subsection*{Limiting mean value}

Theorem \ref{thm:mean_RNG} states that the mean value of the RNG
is given by the following expression:
\begin{equation}
\left\langle d\right\rangle _{n}\left(\sigma\right)=\frac{2\sigma}{\left|\Gamma\right|}\sum_{v\in\mathcal{V}_{\mathcal{R}}}\frac{1}{\deg\left(v\right)}.\label{eq:-98}
\end{equation}
This expression bears obvious similarity to the result introduced
in \cite{RudWigYes_arxiv21} for planar domains
\begin{equation}
\left\langle d\right\rangle _{n}\left(\sigma\right)=\frac{2\left|\partial\Omega\right|}{\text{\ensuremath{\left|\Omega\right|}}}\sigma,\label{eq:-33-2}
\end{equation}
which is also the expression proven in \cite{RudWig_amq21} for the
hemisphere.

In the graph setting, the boundary term $\left|\partial\Omega\right|$
is replaced by a discrete measure on the set of Robin points. Interestingly,
this discrete measure assigns to each vertex a total weight which
is inversely proportional to its degree. Heuristically, if a vertex
``meets'' the graph from many different sides, then it is less likely
to ``feel'' the Robin perturbation.

While this property is seemingly unique for the graph setting, we
believe that at least in a sense, a similar effect also exists for
two-dimensional domains. For instance, one can consider a domain with
a two sided boundary (as in Figure \ref{fig: bndry}) and replace
the usual one-sided Robin condition $\frac{\partial f}{\partial n}+\sigma f=0$
with the two-sided Robin condition:
\begin{equation}
\frac{\partial f_{1}}{\partial n_{1}}+\frac{\partial f_{2}}{\partial n_{2}}+\sigma f=0.\label{eq:-12}
\end{equation}
\begin{figure}
\includegraphics[scale=0.4]{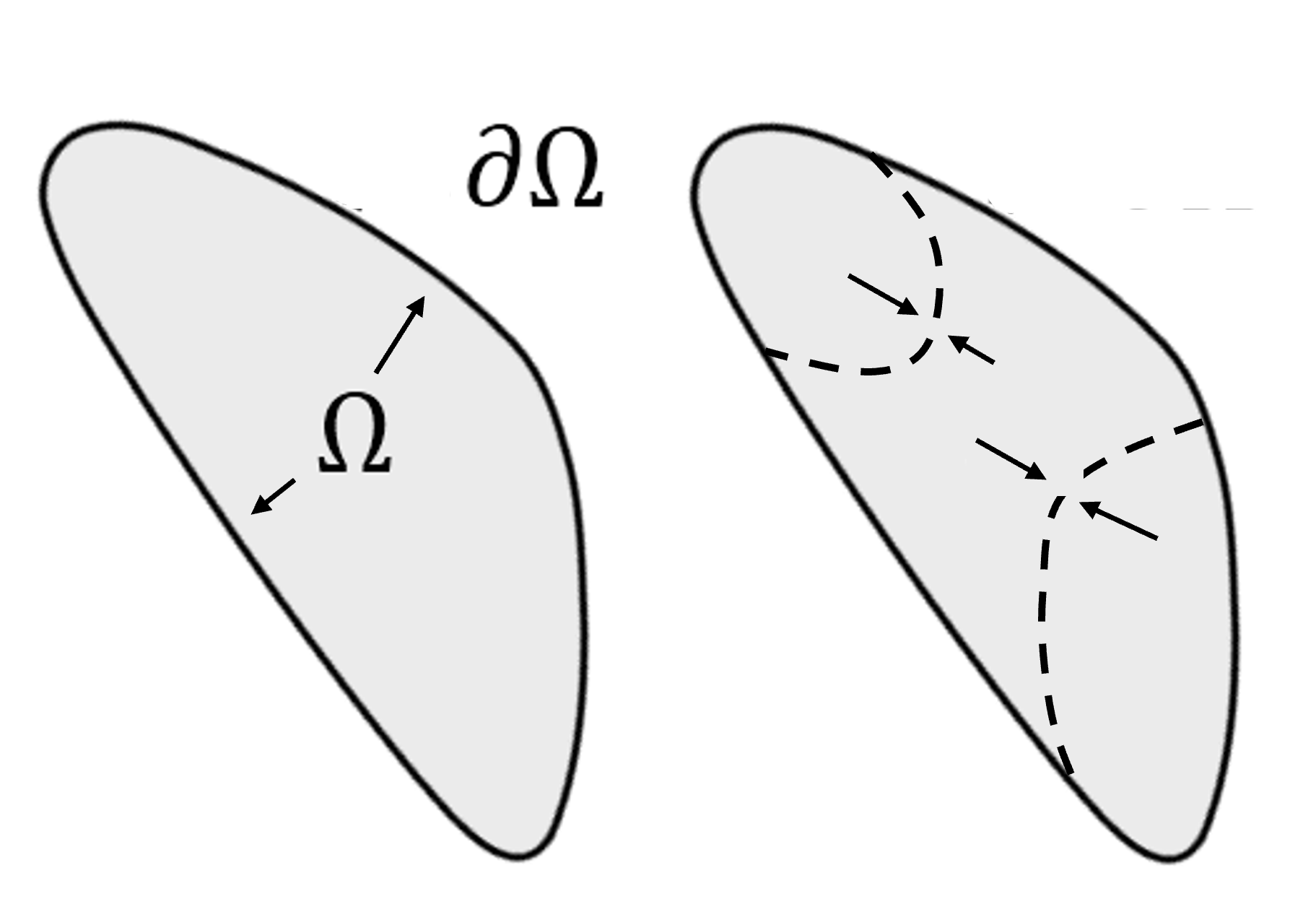}

\caption[Domains with different ``degrees''.]{A one-sided boundary compared to a two-sided boundary. \label{fig: bndry}}
\end{figure}
In this case, we believe that the corresponding expression in (\ref{eq:-33-2})
should be divided by two, whenever the two-sided boundary is considered.
In this sense, the degree of the vertex can be replaced by the number
of sides of the boundary which are in contact with the domain. This
idea is further developed in \cite{BifKer_arXiv}, where the notion
of circumference for a quantum graph is introduced.

\subsection*{Bounds}

The uniform boundedness of the Robin-Neumann gaps in $\sigma$ proven
in Theorems \ref{thm:Lipschitz_RNG} and \ref{thm:explicit_bounds}
agrees with the result proven for star graphs in \cite{RivRoy_jphys20}.
Having uniform bounds on the gaps is interesting, since this does
not always hold when one passes to the two-dimensional setting. While
the sequence of RNG is bounded for some domains (with explicit bounds
for the rectangle given in \cite{Rudnick2021}), the sequence is known
to be unbounded for the hemisphere. It is also conjectured to be unbounded
for certain planar domains, like the disk.

\subsection*{RNG distribution and its higher moments}

Theorems \ref{thm:probability} and \ref{thm:converging_subsequence}
provide information about the distribution of the RNG, which is governed
by a probability distribution similar to the one discussed in \cite{RivRoy_jphys20}
(see e.g. Figure \ref{fig:RNG-hist-1}). A very natural question to
ask is what else can be said about this probability distribution,
and what other geometric information about the graph it holds.\\
Since the expectation of this probability measure can be computed
explicitly (as in Theorem \ref{thm:mean_RNG}), it is sensible to
try and study the measure $\mu_{\sigma}$ by computing the higher
moments as well. Naively, the computation of the higher moments could
be carried out by an approach similar to the one used in proving Theorem
\ref{thm:mean_RNG} -- defining the higher moments as functions on
the secular manifold, and then computing the corresponding integral.
Yet, it turns out that the higher moments cannot be expressed as well
defined functions on the secular manifold. Since this approach fails
for the higher moments, this problem holds an additional challenge
of finding a different way to perform the computation.

\begin{figure}
\includegraphics[scale=0.6]{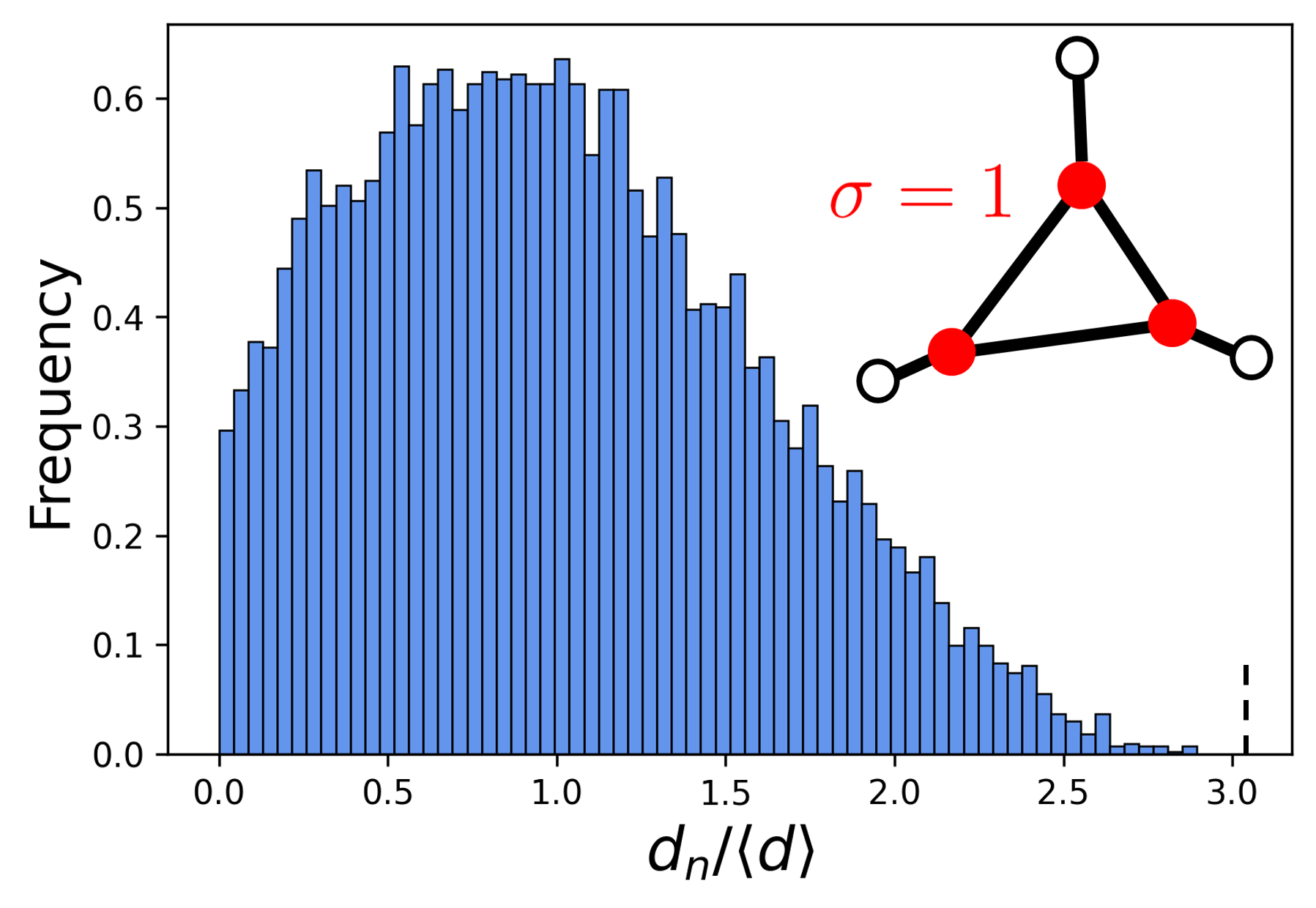}

\caption[Histogram of Robin-Neumann gaps.]{Histogram of the first $10,000$ values of the RNG, normalized so
that $\left\langle d_{n}\right\rangle \left(\sigma\right)=1$. The
Robin vertices are marked in red in the upper right corner. The frequencies
are normalized so that the total area under the histogram is $1$,
indicating the possible probability distribution. An estimate for
the upper bound appearing in (\ref{eq:explicit_bounds}) is marked
by a dashed vertical line.\label{fig:RNG-hist-1}}
\end{figure}

\subsection*{From Robin to other vertex conditions}

Lastly, we address the possible generalization of the results to other
vertex conditions. For instance, one may ask whether similar results
can be obtained for the counterpart $\delta'$-type condition (see
e.g. \cite{AlbeverioGesztesy_solvable,AvrExnLas_prl94}). In this
case, the eigenvalue curves are non-increasing with the coupling parameter
$\sigma$, and one can similarly define the family of gaps.

From numerical exploration, it seems that if one defines
\begin{equation}
\tilde{d}_{n}\left(\sigma\right):=\lambda_{n}\left(0\right)-\lambda_{n}\left(\sigma\right)=k_{n}^{2}\left(0\right)-k_{n}^{2}\left(\sigma\right),\label{eq:-33}
\end{equation}
then the given sequence is no longer bounded, and its mean value does
not converge. This suggests that generalization of the results above
to other vertex conditions might not always be possible. Interestingly,
if one instead defines 
\begin{equation}
\tilde{d}_{n}\left(\sigma\right)=k_{n}\left(0\right)-k_{n}\left(\sigma\right),\label{eq:-33-1}
\end{equation}
then the given sequence does seem to be bounded, and the mean value
converges as before (see Figure \ref{fig:DPRNG}), although it is
not clear to us at this point what it converges to.

\begin{figure}
\includegraphics[scale=0.6]{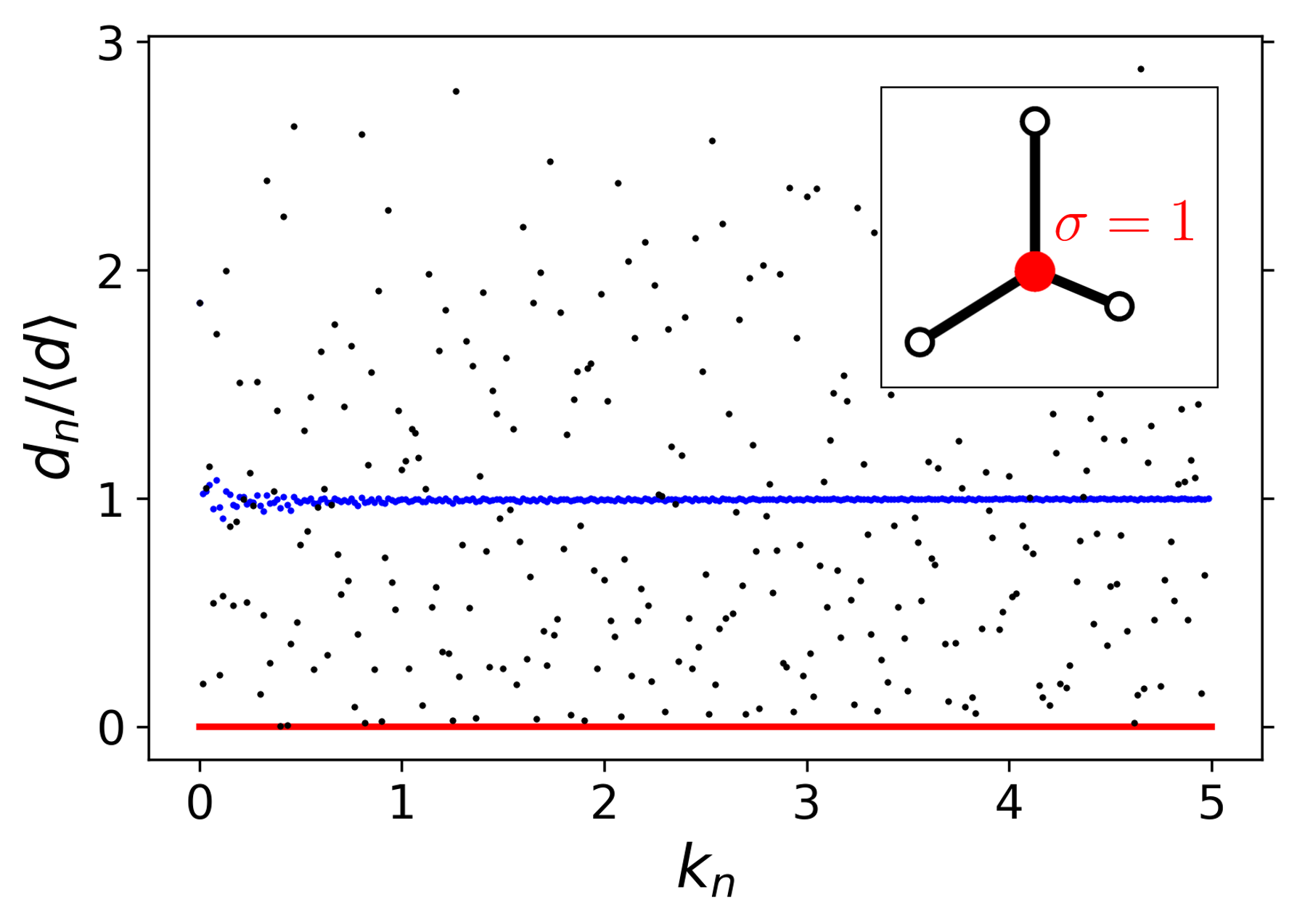}

\caption{The first three-hundred values of the spectral gap for the $\delta'$
condition on a star graph (red points). The mean value of the sequence
up to the $n^{\textrm{th}}$ eigenvalue is indicated by blue points
and suggests that it converges (namely, the Ces\`{a}ro mean exists).
The plot also suggests that the sequence is bounded. \label{fig:DPRNG}}
\end{figure}

When trying to repeat the computation of the mean value of gaps for
the $\delta'$ condition, one comes across a problem similar to the
one with the higher moments. It turns out that it is not possible
to define the corresponding gap as a function on the secular manifold,
which makes it difficult to apply the approach presented in this work.
It is possible that the computation may be done using the local Weyl
law proven in \cite{BifKer_arXiv} for general vertex conditions.
But as the numerical example above suggests, the results might be
very different.
\begin{acknowledgement*}
We are indebted to Uzy Smilansky for gluing our team together. Uzy
participated in many discussions towards the formation of this paper.
He provided ideas which stimulated the scientific content and promoted
our work process. We are grateful to him for all that.

We thank Nadav Yesha for a very clear presentation of the analogous
results on manifolds, which encouraged us to start this work. This
work also gained from interesting discussions and email correspondence
with Lior Alon, Gregory Berkolaiko, James Kennedy and Zeev Rudnick,
and we thank them for that.

RB and GS were supported by ISF (Grant No. 844/19) and by the Binational
Science Foundation Grant (Grant No. 2016281).
\end{acknowledgement*}

\appendix

\section{\label{sec:Appendix-A} Omitting the assumption of independence over
$\mathbb{Q}$}

Recall that in order to apply the ergodic theorem (Theorem \ref{thm:ergodic-theorem}),
we added the assumption that the entries of the vector of edge lengths
$\vec{\ell}$ are linearly independent over $\mathbb{Q}$. We now
show that the result of Theorem \ref{thm:mean_RNG} in fact holds
without this assumption. A similar proof shows that Theorem \ref{thm:Weyl-law}
also holds without this assumption (see Remark \ref{rem: assumption also removed from Weyl law}).
\begin{prop}
Assumption \ref{assumption: rationality} can be omitted in Theorem
\ref{thm:mean_RNG}.
\end{prop}

\begin{proof}
This is a simple denseness argument. Fix $\sigma>0$, a discrete graph
$G=\left(\mathcal{V},\mathcal{E}\right)$ and a Robin set $\mathcal{V}_{R}\subset\mathcal{V}$.
Denote:
\begin{equation}
\mathbb{R}_{+}^{E}=\left\{ \vec{x}\in\mathbb{R}^{E}:x_{i}>0,\forall i\in\left\{ 1,..,E\right\} \right\} .\label{eq:-11}
\end{equation}
For $\vec{\ell}\in\mathbb{R}_{+}^{E}$, denote by $\Gamma_{\vec{\ell}}$
the metric graph obtained by assigning the vector of edge lengths
$\vec{\ell}$ to the fixed combinatorial graph $G$.

Define the following function:
\begin{align}
 & \phi_{1}:\mathbb{R}_{+}^{E}\rightarrow\mathbb{R}^{\mathbb{N}},\label{eq:-10}\\
 & \phi_{1}\left(\vec{\ell}\right)=\left(d_{1}^{\vec{\ell}}\left(\sigma\right),d_{2}^{\vec{\ell}}\left(\sigma\right),...\right),\label{eq:-14}
\end{align}
where $d_{n}^{\vec{\ell}}\left(\sigma\right)$ is the RNG for the
metric graph $\Gamma_{\vec{\ell}}$ with corresponding Robin vertex
set $\mathcal{V}_{R}$.

Denote by $P$ the subset of $\mathbb{R}_{+}^{E}$ of vectors whose
coordinates are rationally independent. This is a dense subset of
$\mathbb{R}_{+}^{E}$. Denote the set of Ces\`{a}ro summable sequences
by $\mathcal{C}$. Define the following functions:
\begin{align}
 & \phi_{2}:\mathcal{C}\rightarrow\mathbb{R},\label{eq:-63-1}\\
 & \phi_{2}\left(\left(c_{n}\right)_{n=1}^{\infty}\right)=\lim_{N\rightarrow\infty}\frac{1}{N}\sum_{n=1}^{N}c_{n}\label{eq:-64}\\
 & \phi:P\rightarrow\mathbb{R},\label{eq:-65}\\
 & \phi=\phi_{2}\circ\left(\phi_{1}|_{P}\right).\label{eq:-66}
\end{align}

By the version we proved for Theorem \ref{thm:mean_RNG}, $\phi$
is a well defined function on $P$. For all $\vec{x}\in\mathbb{R}_{+}^{E}$,
there exists a neighborhood $U\subset\mathbb{R}_{+}^{E}$ of $\vec{x}$,
such that $\phi$ is uniformly continuous on $U\cap P,$ since it
is simply given by the expression:
\begin{equation}
\phi\left(\vec{\ell}\right)=\frac{2\sigma}{\sum_{e=1}^{E}\ell_{e}}\sum_{v\in\VR}\frac{1}{\deg\left(v\right)}.\label{eq:-67}
\end{equation}
Since $P$ is dense in $\mathbb{R}_{+}^{E}$ and $\phi$ is locally
uniformly continuous (in the sense mentioned above), it can be extended
into a continuous function $\tilde{\phi}$ on $\mathbb{R}_{+}^{E}$.

To complete the proof, we should show that $\tilde{\phi}=\phi_{2}\circ\phi_{1}$
and that $\tilde{\phi}$ is given by expression (\ref{eq:-67}). This
is based on standard topological arguments which we merely sketch
here, and refer the interested reader to \cite[prop. 5.16]{Sofer2022}
for further details. First, one shows that $\phi_{1}$ is continuous.
From here follows $\phi_{1}\left(\mathbb{R}_{+}^{E}\right)\subset\mathcal{C}$,
using that set of Ces\`{a}ro summable sequences is closed. Now one
gets that $\phi_{2}\circ\phi_{1}$ is well defined and continuous,
and concludes $\tilde{\phi}=\phi_{2}\circ\phi_{1}$ since those functions
agree on the dense set $P$. Finally, the continuity of $\tilde{\phi}$
implies that it is indeed given by (\ref{eq:-67}) and completes the
proof.
\end{proof}
\begin{rem}
\label{rem: assumption also removed from Weyl law} A similar proof
as above shows that assumption \ref{assumption: rationality} may
be omitted also from Lemmas \ref{lem:Amean} and \ref{lem:Uncorrelation},
and hence also from Theorem \ref{thm:Weyl-law}. The required change
in the proof is to modify the function $\phi_{1}:\mathbb{R}_{+}^{E}\rightarrow\mathbb{R}^{\mathbb{N}}$
accordingly, so that it returns either the values of (\ref{eq:-42-1-1})
(to prove Lemma \ref{lem:Amean}) or the values of (\ref{eq:-43-2})
(to prove Lemma \ref{lem:Uncorrelation}).
\end{rem}

\begin{rem}
It is worth noting that while the limiting mean value of the RNG is
the same for the rationally dependent case, the behavior of the RNG
sequence itself might be drastically different in that case. For instance,
for the case of an equilateral star graph, one can show that the RNG
accumulates around two values, and does not get close to the mean
value. Nevertheless, the Ces\`{a}ro mean converges to its expected
value as in Theorem \ref{thm:mean_RNG}. The above is nicely exemplified
in Figure \ref{fig:equistar-gap}.
\end{rem}

\bibliographystyle{plain}
\bibliography{GlobalBib_210709}

\bigskip

\bigskip

\bigskip

Ram Band and Gilad Sofer

Faculty of Mathematics

Technion -- Israel Institute of Technology

32000 Haifa

Israel

e-mail: ramband@technion.ac.il, gilad.sofer@campus.technion.ac.il

\bigskip

Holgar Schanz

Hochschule Magdeburg-Stendal

39114 Magdeburg

Germany

e-mail: holger.schanz@h2.de
\end{document}